\documentclass{ws-ppl}
\usepackage{cite}
\usepackage{mathtools}
\usepackage{mflogo}
\usepackage{paralist}
\usepackage{textcomp}
\usepackage{url}
\usepackage{color}
\definecolor{ptw}{rgb}{0.21,0.39,0.55}
\definecolor{vec}{rgb}{0.62,0.02,0.03}
\usepackage[final]{listings}
\providecommand{\fma}{\ensuremath\mathop{\mathrm{fma}}}
\providecommand{\fmax}{\ensuremath\mathop{\mathrm{fmax}}}
\providecommand{\fmin}{\ensuremath\mathop{\mathrm{fmin}}}
\providecommand{\hypot}{\ensuremath\mathop{\mathrm{hypot}}}
\providecommand{\invsqrt}{\ensuremath\mathop{\mathrm{invsqrt}}}
\providecommand{\sign}{\ensuremath\mathop{\mathrm{sign}}}
\begin{document}
\markboth{V.~Novakovi\'{c}}{Batched computation of the SVDs of order two by the AVX-512 vectorization}
%
%
\catchline{}{}{}{}{}
%
%
\title{Batched computation of the singular value\\decompositions of order two by the AVX-512 vectorization}
\author{Vedran Novakovi\'{c}\footnote{%
e-mail address: \texttt{venovako@venovako.eu}\hfill\texttt{https://orcid.org/0000-0003-2964-9674}}}
\address{completed a major part of this research as a collaborator on the MFBDA project\\
  10000 Zagreb, Croatia}
\maketitle
\begin{history}
\received{(received date)}
\revised{(revised date)}
\end{history}
\begin{abstract}
  In this paper a vectorized algorithm for simultaneously computing up
  to eight singular value decompositions (SVDs, each of the form
  $A=U\Sigma V^{\ast}$) of real or complex matrices of order two is
  proposed.  The algorithm extends to a batch of matrices of an
  arbitrary length $n$, that arises, for example, in the annihilation
  part of the parallel Kogbetliantz algorithm for the SVD of a square
  matrix of order $2n$.  The SVD algorithm for a single matrix of
  order two is derived first.  It scales, in most instances
  error-free, the input matrix $A$ such that its singular values
  $\Sigma_{ii}$ cannot overflow whenever its elements are finite, and
  then computes the URV factorization of the scaled matrix, followed
  by the SVD of a non-negative upper-triangular middle factor.  A
  vector-friendly data layout for the batch is then introduced, where
  the same-indexed elements of each of the input and the output
  matrices form vectors, and the algorithm's steps over such vectors
  are described.  The vectorized approach is then shown to be about
  three times faster than processing each matrix in isolation, while
  slightly improving accuracy over the straightforward method for the
  $2\times 2$ SVD.
\end{abstract}
\setcounter{footnote}{0}
\renewcommand{\thefootnote}{\alph{footnote}}
\newlength\fbw
\lstloadlanguages{C}
\lstset{language=C,extendedchars=false,numbers=left,numberstyle=\tiny,frame=lines,basicstyle=\small\ttfamily,columns=fullflexible,texcl=false,mathescape=true,lineskip=.25\baselineskip}
%
%
\section{Introduction}\label{s:1}
%
%
Let a finite sequence $\mathbf{A}=(\mathbf{A}^{[k]})_k$, where
$1\le k\le n$, of complex $2\times 2$ matrices be given, and let the
corresponding sequences $\mathbf{U}=(\mathbf{U}^{[k]})_k$,
$\mathbf{V}=(\mathbf{V}^{[k]})_k$ of $2\times 2$ unitary matrices be
sought for, as well as a sequence
$\mathbf{\Sigma}=(\mathbf{\Sigma}^{[k]})_k$ of $2\times 2$ diagonal
matrices with the real and non-negative diagonal elements, such that
$\mathbf{A}^{[k]}=\mathbf{U}^{[k]}\mathbf{\Sigma}^{[k]}(\mathbf{V}^{[k]})^{\ast}$,
i.e., for each $k$, the right hand side of the equation is the
singular value decomposition (SVD) of the left hand side.  This batch
of $2\times 2$ SVD computational tasks arises naturally in, e.g.,
parallelization of the Kogbetliantz algorithm~\cite{Kogbetliantz-55}
for the $2n\times 2n$
SVD~\cite{Becka-Oksa-Vajtersic-02,Oksa-et-al-19,Novakovic-Singer-20}.
A parallel step of the algorithm, repeated until convergence, amounts
to forming and processing such a batch, with each $\mathbf{A}^{[k]}$
assembled column by column from the elements of the iteration matrix
at the suitably chosen pivot positions $(p_k,p_k)$, $(q_k,p_k)$,
$(p_k,q_k)$, and $(q_k,q_k)$.  The iteration matrix is then updated
from the left by $(\mathbf{U}^{[k]})^{\ast}$ and from the right by
$\mathbf{V}^{[k]}$, transforming the $p_k$th and the $q_k$th rows and
columns, respectively, while annihilating the off-diagonal pivot
positions.

For each $k$, the matrices $\mathbf{A}^{[k]}$, $\mathbf{U}^{[k]}$,
$\mathbf{V}^{[k]}$, and $\mathbf{\Sigma}^{[k]}$ have the following
elements,
\begin{displaymath}
  \mathbf{A}^{[k]}=
  \begin{bmatrix}
    a_{11}^{[k]} & a_{12}^{[k]}\\
    a_{21}^{[k]} & a_{22}^{[k]}
  \end{bmatrix},\hfill
  \mathbf{U}^{[k]}=
  \begin{bmatrix}
    u_{11}^{[k]} & u_{12}^{[k]}\\
    u_{21}^{[k]} & u_{22}^{[k]}
  \end{bmatrix},\hfill
  \mathbf{V}^{[k]}=
  \begin{bmatrix}
    v_{11}^{[k]} & v_{12}^{[k]}\\
    v_{21}^{[k]} & v_{22}^{[k]}
  \end{bmatrix},\hfill
  \mathbf{\Sigma}^{[k]}=
  \begin{bmatrix}
    \sigma_{\max}^{[k]} & 0\\
    0 & \sigma_{\min}^{[k]}
  \end{bmatrix},
\end{displaymath}
\addtocounter{equation}{-1}
where $\sigma_{\max}^{[k]}\ge\sigma_{\min}^{[k]}\ge 0$.  When its
actual index $k$ is either implied or irrelevant, $\mathbf{A}^{[k]}$
is denoted by $A$.  Similarly, $U$, $V$, and $\Sigma$ denote
$\mathbf{U}^{[k]}$, $\mathbf{V}^{[k]}$, and $\mathbf{\Sigma}^{[k]}$,
respectively, in such a case, and the bracketed indices of the
particular elements are also omitted.

When computing in the machine's floating-point arithmetic, the real
and the imaginary parts of the input elements are assumed to be
rounded to finite (i.e., excluding $\pm\infty$ and \texttt{NaN})
\texttt{double} precision quantities, but the SVD computations can
similarly be vectorized in single precision (\texttt{float} datatype
in the C language~\cite{C-18}).

Let \texttt{C} and \texttt{W} denote the CPU's cache line size and the
maximal SIMD width, both expressed in bytes, respectively.  For an
Intel CPU with the 512-bit Advanced Vector Extensions Foundation
(AVX-512F) instruction set~\cite{Intel-19},
$\mathtt{C}=\mathtt{W}=64$.  Let \texttt{B} be the size in bytes of
the chosen underlying datatype \texttt{T} (here,
$\mathtt{T}=\mathtt{double}$ in the real and the complex case alike,
so $\mathtt{B}=8$), and let $\mathtt{S}=\mathtt{W}/\mathtt{B}=8$.  If
$n\bmod\mathtt{S}\ne 0$, let
$\hat{n}=n+(\mathtt{S}-(n\bmod\mathtt{S}))$, else let $\hat{n}=n$.

This paper aims to show how to single-threadedly compute as many SVDs
at the same time as there are the SIMD/vector lanes available
(\texttt{S}), one SVD by each lane.  Furthermore, these vectorized
computations can execute concurrently on the non-overlapping batch
chunks assigned to the multiple CPU cores.

Techniques similar to the ones proposed in this paper have already
been applied in~\cite{Singer-DiNapoli-Novakovic-Caklovic-19} for
vectorization of the Hari--Zimmermann joint diagonalizations of a
complex positive definite pair of matrices~\cite{Hari-19} of order
two, and could be, as future work, for the real variant of the
Hari--Zimmermann algorithm for the generalized
eigendecomposition~\cite{Hari-18} and the generalized
SVD~\cite{Novakovic-Singer-Singer-15}.  Those efforts do not use the C
compiler intrinsics, but rely instead on the vectorization capabilities
of the Intel Fortran compiler over data laid out in a vector-friendly
fashion similar to the one described in section~\ref{s:3}.  Simple as
it may seem, it is also a more fragile way of expressing the vector
operations, should the compiler ever renegade on the present behavior
of its autovectorizer.  The intrinsics approach has been tried
in~\cite{Novakovic-17} with 256-bit-wide vectors of the
AVX2+FMA~\cite{Intel-19} instruction set, alongside AVX-512F, for
vectorization of the eigendecompositions of symmetric matrices of
order two by the Jacobi rotations computed similarly
to~\cite{Drmac-97}.  This way the one-sided Jacobi SVD (and,
similarly, the hyperbolic SVD) of real $n\times n$ matrices can be
significantly sped up when $n$ is small enough to make the
eigendecompositions' execution time comparable to the $2\times 2$
Grammian formations and the column updates, e.g., when the targeted
matrices are the block pivots formed in a block-Jacobi
algorithm~\cite{Hari-SingerSanja-SingerSasa-10,Hari-SingerSanja-SingerSasa-14}.

In numerical linear algebra the term ``batched computation'' is
well-established, signifying a simultaneous processing of a large
quantity of relatively small problems, e.g., the LU and the Cholesky
factorizations~\cite{Haidar-et-al-18} and the corresponding linear
system solving~\cite{Dongarra-et-al-18} on the GPUs, with appropriate
data layouts.  It is therefore both justifiable and convenient to
reuse the term in the present context.

\looseness=-1
This paper is organized as follows.  A non-vectorized Kogbetliantz
method for the SVD of a matrix of order two is presented in
section~\ref{s:2}.  In section~\ref{s:3} a vector-friendly data layout
is proposed, followed by a summary of the vectorized algorithm for the
batched $2\times 2$ SVDs in section~\ref{s:4}.  The algorithm
comprises the following phases:
\begin{compactenum}[1.]
  \item scaling the input matrices $\mathbf{A}$, each by a suitable
    power of two, to avoid any overflows and most underflows,
    vectorized in section~\ref{s:5} and based on
    subsection~\ref{ss:2.1},
  \item the URV factorizations~\cite{Stewart-92}, vectorized in
    section~\ref{s:6} and based on subsection~\ref{ss:2.2}, of the
    matrices from the previous phase, into the real, non-negative,
    upper-triangular middle factors $\mathbf{R}$ and the unitary left
    and right factors,
  \item the singular value decompositions of the factors $\mathbf{R}$
    from the previous phase, vectorized in section~\ref{s:7} and based
    on subsection~\ref{ss:2.3}, yielding the scaled $\mathbf{\Sigma}$,
    and
  \item assembling of the left ($\mathbf{U}$) and the right
    ($\mathbf{V}$) singular vectors of the matrices $\mathbf{A}$.
\end{compactenum}
The numerical testing results follow in section~\ref{s:8}, and the
conclusions in section~\ref{s:9}.
%
%
\section{The Kogbetliantz algorithm for the SVD of order two}\label{s:2}
%
%
The pointwise, non-vectorized Kogbetliantz algorithm for the SVD of a
matrix of order two has been an active subject of
research~\cite{Charlier-Vanbegin-VanDooren-87,Hari-Matejas-09,Matejas-Hari-15},
and has been implemented for real matrices in
LAPACK's~\cite{Anderson-et-al-99} \texttt{xLASV2} (for the full SVD)
and \texttt{xLAS2} (for the singular values only) routines, where
$\mathtt{x}\in\{\mathtt{S},\mathtt{D}\}$.  Here a simplified version
of the algorithm from~\cite[trigonometric case]{Novakovic-Singer-20}
is described, with an early reduction of a complex matrix to the real
one that is partly influenced by, but improves on,~\cite{Qiao-Wang-02}.

It is assumed in the paper that the floating-point
arithmetic~\cite{IEEE-754-2008} is nonstop, i.e., does not trap on
exceptions, and has the gradual underflow, i.e.,
Flush-denormals-To-Zero (FTZ) and Denormals-Are-Zero (DAZ) processor
flags~\cite{Intel-19} are disabled.

To compute $|z|$ and $e^{\mathrm{i}\arg{z}}$ for a complex $z$ with
both components finite, including $z=0$, while avoiding the complex
arithmetic operations, use the $\hypot(a,b)=\sqrt{a^2+b^2}$
function~\cite{C-18}.  With \texttt{DBL\_TRUE\_MIN} being the smallest
positive non-zero (and thus subnormal, or denormal in the old
parlance) double precision value, let
\begin{equation}
  \begin{aligned}
    |z|&=\hypot(\Re{z},\Im{z}),\quad
    e^{\mathrm{i}\arg{z}}=\cos(\arg{z})+\mathrm{i}\cdot\sin(\arg{z}),\\
    \cos(\arg{z})&=\fmin\left(\frac{|\Re{z}|}{|z|},1\right)\cdot\sign{\Re{z}},\quad
    \sin(\arg{z})=\frac{\Im{z}}{\max(|z|,\mathtt{DBL\_TRUE\_MIN})}.
  \end{aligned}
  \label{e:1}
\end{equation}
Here and in the following, $\fmin$ and $\fmax$ are the functions
similar to the ones in the C language~\cite{C-18}, but with a bit
relaxed semantics, that return the minimal (respectively, maximal) of
their two non-\texttt{NaN} arguments, or the second argument if the
first is a \texttt{NaN}, as it is the case with the vector minimum and
maximum~\cite[\texttt{VMINPD} and \texttt{VMAXPD}]{Intel-19}.  See
also~\cite[subsection~6.2]{Singer-DiNapoli-Novakovic-Caklovic-19} for
a similar exploitation of the \texttt{NaN} handling of $\min$ and
$\max$ operations.  It now follows that, when $|z|=0$, and so
$\Re{z}=\Im{z}=0$,
\begin{displaymath}
  \cos(\arg{z})=\sign{\Re{z}}=\pm 1,\quad
  \sin(\arg{z})=0\cdot\sign{\Im{z}}=\pm 0.
\end{displaymath}
\addtocounter{equation}{-1}
The signs of $\Re{z}$ and $\Im{z}$ are thus preserved in
$\cos(\arg{z})$ and $\sin(\arg{z})$, respectively.

A mandatory property of $\hypot$ for~\eqref{e:1} to be applicable to
all inputs $z$, where $\Re{z}$ and $\Im{z}$ are of sufficiently small
magnitude (see subsection~\ref{ss:2.1}), is to have
\begin{displaymath}
  \hypot(a,b)=0\iff a=b=0.
\end{displaymath}
\addtocounter{equation}{-1}

A vectorized $\hypot$ implementation is accessible from the Intel
Short Vector Math Library (SVML) via a compiler intrinsic, as well as
it is a reciprocal square root ($\invsqrt{x}=1/\sqrt{x}$) vectorized
routine, helpful for the cosine calculations in~\eqref{e:7},
\eqref{e:15}, and \eqref{e:16}, though neither is always correctly
rounded to at most half
ulp\footnote{Consult the reports on the High Accuracy functions at
  \url{https://software.intel.com/content/www/us/en/develop/documentation/mkl-vmperfdata/top/real-functions/root.html}
URL\@.}.
%
%
\subsection{Exact scalings of the input matrix}\label{ss:2.1}
%
%
Even if both components of $z$ are finite, $|z|$ from~\eqref{e:1} can
overflow, but $|2^{-1}z|$ cannot.  Scaling a floating-point number by
an integer power of two is exact, except when the significand of a
subnormal result loses a trailing non-zero part due to shifting of the
original significand to the left, or when the result overflows.
Therefore, such scaling~\cite[\texttt{VSCALEFPD}]{Intel-19} is the
best remedy for the absolute value overflow problem.

Let the exponent of a floating-point value (assuming the radix two) be
defined as $\exp_2{0}=-\infty$ and $\exp_2{a}=\lfloor\lg|a|\rfloor$
for a finite non-zero $a$ (see~\cite[\texttt{VGETEXPPD}]{Intel-19}).
Let $h=\mathtt{DBL\_MAX\_EXP}-3$ be two less than the largest
exponent of a finite double precision number.  To find a scaling
factor $2^s$ for $A$, take $s$ as
\begin{equation}
  s=\min\{\mathtt{DBL\_MAX},\min{E^{\Re}},\min{E^{\Im}}\},
  \label{e:2}
\end{equation}
where $E^{\Re}=\{E_{ij}^{\Re}\mid 1\le i,j\le 2\}$ and
$E^{\Im}=\{E_{ij}^{\Im}\mid 1\le i,j\le 2\}$ are computed as
\begin{displaymath}
  E_{ij}^{\Re}=h-\exp_2^{}\Re{a_{ij}^{}},\quad
  E_{ij}^{\Im}=h-\exp_2^{}\Im{a_{ij}^{}}.
\end{displaymath}
\addtocounter{equation}{-1}
Note that $-2\le s\le\mathtt{DBL\_MAX}$, due to the definition of $h$.
If $A$ is real, $E^{\Im}$ is not used.
The upper bound on $s$ is required to be finite, since
$0\cdot 2^{\infty}$ would result in a \texttt{NaN}.

If there is a value of a huge magnitude (i.e., with its exponent
greater than $h$) in $A$, $s$ from~\eqref{e:2} will be negative and
the huge values will decrease, either twofold or fourfold.  Else, $s$
will be the maximal non-negative amount by which the exponents of the
values in $A$ can jointly be increased, thus taking the very small
values out of the subnormal range if possible, without any of the new
exponents going over $h$.

Let $\widehat{A}=2^s A$, and let $\hat{a}_j$ denote the $j$th column
of $\widehat{A}$.  The Frobenius norm of $\hat{a}_j$,
\begin{equation}
  \|\hat{a}_j\|_F=\hypot(|\hat{a}_{1j}|,|\hat{a}_{2j}|).
  \label{e:3}
\end{equation}
cannot overflow (see~\eqref{e:27} in the proof of Theorem~\ref{t:1} in
subsection~\ref{ss:2.3}).

This scaling is both a simplification and an improvement
of~\cite[subsection~2.3.2]{Novakovic-Singer-20}, which guarantees that
the computed scaled singular values are finite, while avoiding any
branching, lane masking, or recomputing when vectorized, with the only
adverse effect being a potential sacrifice of the tiniest subnormal
values in the presence of a huge one (i.e., with its exponent strictly
greater than $h$) in $A$.
%
%
\subsection{The URV factorization of a well-scaled matrix}\label{ss:2.2}
%
%
If $\|\hat{a}_1\|_F<\|\hat{a}_2\|_F$, let
$P_c=[\begin{smallmatrix}0&1\\1&0\end{smallmatrix}]$, else let
$P_c=[\begin{smallmatrix}1&0\\0&1\end{smallmatrix}]$.  Denote the
column-pivoted $\widehat{A}$ by $A'=\widehat{A}P_c$.  If
$|a_{11}'|<|a_{21}'|$, let
$P_r^{\ast}=[\begin{smallmatrix}0&1\\1&0\end{smallmatrix}]$, else let
$P_r^{\ast}=[\begin{smallmatrix}1&0\\0&1\end{smallmatrix}]$.  Denote
the row-sorted $A'$ by $A''=P_r^{\ast} A'$.  To make $a_1''$ real and
non-negative, let
\begin{equation}
  D^{\ast}=
  \begin{bmatrix}
    e^{-\mathrm{i}\arg{a_{11}''}} & 0\\
    0 & e^{-\mathrm{i}\arg{a_{21}''}}
  \end{bmatrix},\quad
  A'''=D^{\ast} A''.
  \label{e:4}
\end{equation}

Complex multiplication, required in~\eqref{e:4}, \eqref{e:8},
\eqref{e:9}, and \eqref{e:11}, is performed using the fused
multiply-add operations with a single rounding~\cite{C-18},
$\fma(a,b,c)=a\cdot b+c$, as in~\cite[\texttt{cuComplex.h}]{NVidia-19}
and~\cite[subsection~3.2.1]{Novakovic-Singer-19}, i.e., for a complex
$c=a\cdot b$ holds
\begin{equation}
  \Re{c}=\fma(\Re{a},\Re{b},-\Im{a}\cdot\Im{b}),\quad
  \Im{c}=\fma(\Re{a},\Im{b},\Im{a}\cdot\Re{b}).
  \label{e:5}
\end{equation}

To annihilate $a_{21}'''$, compute the Givens rotation
$Q_{\alpha}^{\ast}$, where $-\pi/4\le\alpha\le 0$, as 
\begin{equation}
  Q_{\alpha}^{\ast}=\cos\alpha
  \begin{bmatrix}
    1 & -\tan\alpha\\
    \tan\alpha & 1
  \end{bmatrix},\quad
  R''=Q_{\alpha}^{\ast}
  \begin{bmatrix}
    a_1''' & a_2'''
  \end{bmatrix}=
  \begin{bmatrix}
    r_{11}^{} & r_{12}'\\
    0 & r_{22}''
  \end{bmatrix},\quad
  r_{11}=\|a_1'''\|_F.
  \label{e:6}
\end{equation}
Since the column norms of a well-scaled $\widehat{A}$ are finite, its
column-pivoted, row-sorted QR factorization in~\eqref{e:6} cannot
result in an infinite element in $R''$.

If $a_{11}'''=0$ then $A'''=0$ as a special case.  Handling special
cases in a vectorized way is difficult as it implies branching or
using the instructions with a lane mask.  However, $\fmax$ function
aids in avoiding both of these approaches similarly to $\fmin$
in~\eqref{e:1}, since $\tan\alpha$ and $\cos\alpha$ from~\eqref{e:6}
can be computed as
\begin{equation}
  \tan\alpha=-\fmax(a_{21}'''/a_{11}''',0),\quad
  \cos\alpha=\invsqrt(\fma(\tan\alpha,\tan\alpha,1)).
  \label{e:7}
\end{equation}

To make $r_{12}'$ from~\eqref{e:5} real (see~\cite{Qiao-Wang-02}) and
non-negative, take $\widetilde{D}$ and obtain $R'$ as
\begin{equation}
  \widetilde{D}=
  \begin{bmatrix}
    1 & 0\\
    0 & e^{-\mathrm{i}\arg{r_{12}'}}
  \end{bmatrix},\quad
  R'=R''\widetilde{D}=
  \begin{bmatrix}
    r_{11}^{} & r_{12}^{}\\
    0 & r_{22}'
  \end{bmatrix},\quad
  r_{12}^{}\ge 0.
  \label{e:8}
\end{equation}
Similarly, to make $r_{22}'$ from~\eqref{e:8} real and non-negative,
take $\widehat{D}^{\ast}$ and obtain $R$ as
\begin{equation}
  \widehat{D}^{\ast}=
  \begin{bmatrix}
    1 & 0\\
    0 & e^{-\mathrm{i}\arg{r_{22}'}}
  \end{bmatrix},\quad
  R=\widehat{D}^{\ast}R'=
  \begin{bmatrix}
    r_{11}^{} & r_{12}^{}\\
    0 & r_{22}^{}
  \end{bmatrix},\quad
  r_{11}^{}\ge\max\{r_{12}^{},r_{22}^{}\},
  \label{e:9}
\end{equation}
due to the column pivoting.  Specifically, if $A$ is already real,
then
\begin{equation}
  D^{\ast}=
  \begin{bmatrix}
    \sign{a_{11}''} & 0\\
    0 & \sign{a_{21}''}
  \end{bmatrix},\quad
  \widetilde{D}=
  \begin{bmatrix}
    1 & 0\\
    0 & \sign{r_{12}'}
  \end{bmatrix},\quad
  \widehat{D}^{\ast}=
  \begin{bmatrix}
    1 & 0\\
    0 & \sign{r_{22}'}
  \end{bmatrix}.
  \label{e:10}
\end{equation}

Note that, from~\eqref{e:4}, \eqref{e:6}, \eqref{e:8}, and
\eqref{e:9},
\begin{equation}
  R=U_+^{\ast}\widehat{A}V_+^{},\quad
  U_+^{\ast}=\widehat{D}^{\ast}Q_{\alpha}^{\ast}D^{\ast}P_r^{\ast},\quad
  V_+^{}=P_c^{}\widetilde{D},
  \label{e:11}
\end{equation}
where $U_+^{\ast}$ and $V_+^{}$ are unitary, i.e.,
$U_+^{} R V_+^{\ast}$ is a specific URV
factorization~\cite{Stewart-92} of $\widehat{A}$.
%
%
\subsection{The SVD of a special upper-triangular non-negative matrix}\label{ss:2.3}
%
%
Here the plane rotations $U_{\varphi}^{\ast}$ and $V_{\psi}^{}$ are
computed, such that $U_{\varphi}^{\ast} R V_{\psi}^{}=\Sigma'$, where
\begin{equation}
  U_{\varphi}^{\ast}=\cos\varphi
  \begin{bmatrix}
    1 & -\tan\varphi\\
    \tan\varphi & 1
  \end{bmatrix},\quad
  V_{\psi}^{}=\cos\psi
  \begin{bmatrix}
    1 & \tan\psi\\
    -\tan\psi & 1
  \end{bmatrix},\quad
  \Sigma'=
  \begin{bmatrix}
    \sigma_{11}' & 0\\
    0 & \sigma_{22}'
  \end{bmatrix},
  \label{e:12}
\end{equation}
with $R$ from~\eqref{e:9} and $\min\{\sigma_{11}',\sigma_{22}'\}\ge 0$.

Let, as in~\cite[subsection~2.2.1]{Novakovic-Singer-20}, where the
following formulas have been derived,
\begin{equation}
  x=\fmax(r_{12}/r_{11},0),\quad
  y=\fmax(r_{22}/r_{11},0).
  \label{e:13}
\end{equation}
With $x$ and $y$ from~\eqref{e:13}, $0\le x,y\le 1$, compute
\begin{equation}
  \tan(2\varphi)=-\min\left\{\fmax\left(\frac{(2\min(x,y))\max(x,y)}{\fma(x-y,x+y,1)},0\right),\sqrt{\mathtt{DBL\_MAX}}\right\},
  \label{e:14}
\end{equation}
as justified in the next paragraph.

Since the quotient in~\eqref{e:14} is non-negative (when defined),
$\tan(2\varphi)$ is non-positive, and thus $-\pi/4\le\varphi\le 0$.
From $\tan(2\varphi)$ compute
\begin{equation}
  \tan\varphi=\frac{\tan(2\varphi)}{1+\sqrt{\fma(\tan(2\varphi),\tan(2\varphi),1)}},\quad
  \cos\varphi=\invsqrt(\sec^2\varphi),
  \label{e:15}
\end{equation}
with $-1\le\tan\varphi\le 0$ and
$\sec^2\varphi=\fma(\tan\varphi,\tan\varphi,1)$.  Assume that
$\tan(2\varphi)$ was not bounded in magnitude.  If
$|\tan(2\varphi)|=\infty$ in floating-point (this occurs rarely, when
$x>0$, $y=1$, and $x\pm y=\pm 1$), then $\tan\varphi=\mathtt{NaN}$
instead of the correct result, $-1$.  Else, if
$|\tan(2\varphi)|>\sqrt{\mathtt{DBL\_MAX}}$, adding one to its square
would have made little difference before the rounding (and the sum
would have overflown after it), so the square root in~\eqref{e:15}
could be approximated by $|\tan(2\varphi)|$.  Again, with
$|\tan(2\varphi)|$ so obtained, adding one to it in the denominator
in~\eqref{e:15} would have been irrelevant, and $\tan\varphi$ would
have then equaled to $-1$.  Bounding $|\tan(2\varphi)|$ from above as
in~\eqref{e:14} therefore avoids the argument of the square root
overflowing (so using $\hypot(\tan(2\varphi),1)$ instead of
$\sqrt{\fma(\tan(2\varphi),\tan(2\varphi),1)}$ is not required), and
ensures $\tan\varphi=-1$ for all $\tan(2\varphi)$ that would otherwise
be greater than the bound.

Having thus computed $U_{\varphi}$, the right plane rotation
$V_{\psi}$ is constructed from
\begin{equation}
  \tan\psi=\fma(y,\tan\varphi,-x),\quad
  \cos\psi=\invsqrt(\sec^2\psi),
  \label{e:16}
\end{equation}
where $\tan\psi\le 0$ and $\sec^2\psi=\fma(\tan\psi,\tan\psi,1)$.

The following Theorem~\ref{t:1} shows that the special form of $R$
contributes to an important property of the computed scaled singular
values $\Sigma'$; namely, they are already sorted non-ascendingly, and
thus never have to be swapped in a postprocessing step.  Also, the
scaled singular values are always finite in floating-point arithmetic.
\begin{theorem}\label{t:1}
  For $\Sigma'$ it holds $\infty>\sigma_{11}'\ge\sigma_{22}'\ge 0$,
  where
  \begin{equation}
    \sigma_{11}'=(\cos\varphi\cos\psi\sec^2\psi)r_{11}^{},\quad
    \sigma_{22}'=(\cos\varphi\cos\psi\sec^2\varphi)r_{22}^{},
    \label{e:17}
  \end{equation}
  and $\sqrt{2}\ge|\tan\psi|\ge|\tan\varphi|\ge 0$.
\end{theorem}
\begin{proof}
  Assume $r_{12}^{}>0$ in~\eqref{e:9}, i.e., $x>0$ in~\eqref{e:13}.
  Else, from~\eqref{e:13}, \eqref{e:14}, and \eqref{e:16}, both
  tangents are zero and both cosines are unity, thus from~\eqref{e:12}
  and~\eqref{e:9} follows
  $\sigma_{11}'=r_{11}^{}\ge r_{22}^{}=\sigma_{22}'$, as claimed, and
  only $\sigma_{11}'<\infty$ remains to be proven.

  From~\eqref{e:15} $\cos\varphi\ne 0$, and from~\eqref{e:16}
  $\tan\psi\ne-\infty$, so $\cos\psi\ne 0$.  Scaling
  $U_{\varphi}^{\ast}$ by $1/\cos\varphi$ and $V_{\psi}^{}$ by
  $1/\cos\psi$ in~\eqref{e:12}, and $R$ by $1/r_{11}$ in~\eqref{e:9},
  from~\eqref{e:13} follows
  \begin{equation}
    \begin{bmatrix}
      1 & -\tan\varphi\\
      \tan\varphi & 1
    \end{bmatrix}
    \begin{bmatrix}
      1 & x\\
      0 & y
    \end{bmatrix}
    \begin{bmatrix}
      1 & \tan\psi\\
      -\tan\psi & 1
    \end{bmatrix}=
    \begin{bmatrix}
      \sigma_{11}'' & 0\\
      0 & \sigma_{22}''
    \end{bmatrix}.
    \label{e:18}
  \end{equation}

  Multiplying the matrices on the left hand side of~\eqref{e:18} and
  equating the elements of the result with the corresponding elements
  of the right hand side, one obtains
  \begin{equation}
    \sigma_{11}''=\tan^2\psi+1=\sec^2\psi,\quad
    \sigma_{22}''=(\tan^2\varphi+1)y=(\sec^2\varphi)y,
    \label{e:19}
  \end{equation}
  after an algebraic simplification using the relation~\eqref{e:16}
  for $\tan\psi=y\tan\varphi-x$.  The equations for $\sigma_{11}'$ and
  $\sigma_{22}'$ from~\eqref{e:17} then follow by multiplying the
  equations for $\sigma_{11}''$ and $\sigma_{22}''$ from~\eqref{e:19},
  respectively, by $r_{11}^{}\cos\varphi\cos\psi$.  Specially,
  $\sigma_{11}'>0$, since $\sigma_{11}'=0$ would imply an obvious
  contradiction $r_{11}^{}=0\ge r_{12}>0$ with the assumption.

  If $y=0$, from~\eqref{e:19} and~\eqref{e:17} it follows
  $\sigma_{11}'>0=\sigma_{22}''=\sigma_{22}'$, and, due
  to~\eqref{e:14}, \eqref{e:15}, and \eqref{e:16}, it holds
  $|\tan\psi|=x>0=|\tan\varphi|$.  If $y=1$ then $x=0$
  from~\eqref{e:9} and~\eqref{e:13}, contrary to the assumption.
  Therefore, $0<y<1$ in the following, and let
  \begin{equation}
    a=-2xy<0,\quad
    b=1+x^2-y^2>0,\quad
    c=b+\sqrt{a^2+b^2}>0.
    \label{e:20}
  \end{equation}

  Then, rewrite $\tan(2\varphi)$ from~\eqref{e:14} using~\eqref{e:20} as
  \begin{equation}
    \tan(2\varphi)=\frac{a}{b},\quad
    \tan^2(2\varphi)+1=\frac{a^2+b^2}{b^2},
    \label{e:21}
  \end{equation}
  as well as $\tan\varphi$ from~\eqref{e:15} using~\eqref{e:21}
  and~\eqref{e:20} as
  \begin{equation}
    \tan\varphi=\frac{\tan(2\varphi)}{1+\sqrt{\tan^2(2\varphi)+1}}=\frac{a/b}{(b+\sqrt{a^2+b^2})/b}=\frac{a}{b+\sqrt{a^2+b^2}}=\frac{a}{c}.
    \label{e:22}
  \end{equation}
  From~\eqref{e:22}, $|\tan\varphi|=|a|/c>0$, what gives
  $|\tan\psi|=y|a|/c+x$ with~\eqref{e:13} and~\eqref{e:16}.  Taking
  the ratio of these two absolute values, it has to be proven that
  \begin{equation}
    \frac{|\tan\psi|}{|\tan\varphi|}=\frac{|a|y+cx}{|a|}\ge 1.
    \label{e:23}
  \end{equation}

  Expanding~\eqref{e:23} using~\eqref{e:20}, it follows
  \begin{equation}
    \frac{|a|y}{|a|}+\frac{cx}{|a|}=y+\frac{(1+x^2-y^2+\sqrt{a^2+b^2})x}{2xy}=\frac{1+x^2+y^2+\sqrt{a^2+b^2}}{2y},
    \label{e:24}
  \end{equation}
  where the argument of the square root can be expressed as
  \begin{equation}
    a^2+b^2=(1+x^2)^2+2(x^2-1)y^2+y^4,
    \label{e:25}
  \end{equation}
  after substitution of~\eqref{e:20} for $a$ and $b$ and a subsequent
  algebraic simplification.  For a fixed but arbitrary
  $y$, \eqref{e:25}, and thus the numerator of~\eqref{e:24}, decrease
  monotonically as $x\to 0$.  Substituting zero for $x$
  in~\eqref{e:24} and~\eqref{e:25}, the former becomes
  \begin{displaymath}
    \frac{1+y^2+\sqrt{(1-y^2)^2}}{2y}=\frac{2}{2y}=\frac{1}{y}>1,
  \end{displaymath}
  \addtocounter{equation}{-1}
  what proves the inequality between the tangents.

  The inequality between the scaled singular values follows easily
  from~\eqref{e:19} as
  \begin{displaymath}
    \frac{\sigma_{22}'}{\sigma_{11}'}=\frac{\sigma_{22}''}{\sigma_{11}''}=\frac{\tan^2\varphi+1}{\tan^2\psi+1}y<1,
  \end{displaymath}
  \addtocounter{equation}{-1}
  since $\tan^2\psi\ge\tan^2\varphi$ and $y<1$.  It remains to be
  shown that $\sigma_{11}'<\infty$ for all $R$ from~\eqref{e:9}.
  If $R$ has not been computed (what is never the case in the proposed
  method) from a well-scaled $\widehat{A}$ (see
  subsection~\ref{ss:2.1}), then even $\sigma_{22}'$ can overflow.

  Observe that $1/\sqrt{2}<\cos\varphi\le 1$ from~\eqref{e:14}, and
  $0<\cos\psi\le\cos\varphi$.  From~\eqref{e:17},
  \begin{equation}
    \frac{\sigma_{11}'}{r_{11}^{}}=\frac{\cos\varphi}{\cos\psi}\le\frac{1}{\cos\psi}.
    \label{e:26}
  \end{equation}
  Since, from~\eqref{e:15} and~\eqref{e:16},
  \begin{displaymath}
    |\tan\psi|=y|\tan\varphi|+x\le y+x,\quad
    \cos\psi=1/\sqrt{1+\tan^2\psi},
  \end{displaymath}
  \addtocounter{equation}{-1}
  $x+y$ has to be bounded from above, to be able to bound $\cos\psi$
  from below, and thus~\eqref{e:26} from above.  From~\eqref{e:9} and the column pivoting goal,
  $r_{11}^2\ge r_{12}^2+r_{22}^2$, what
  gives $x^2+y^2\le 1$ after dividing by $r_{11}^2$, i.e., $x$ and $y$
  are contained in the intersection of the first quadrant and the unit
  disc.  On this domain, $x+y$ attains the maximal value of $\sqrt{2}$
  for $x=y=1/\sqrt{2}$, so $|\tan\psi|\le\sqrt{2}$, as claimed, and
  thus $\cos\psi\ge 1/\sqrt{3}$.  Substituting this lower bound for
  $\cos\psi$ in~\eqref{e:26}, it follows
  $\sigma_{11}'\le\sqrt{3}\cdot r_{11}^{}\ll 2\cdot r_{11}^{}$.

  From subsection~\ref{ss:2.1} and~\eqref{e:6}, since
  \begin{displaymath}
    \max_{1\le i,j\le 2}\{|\Re{\hat{a}_{ij}}|,|\Im{\hat{a}_{ij}}|\}<2^{h+1},
  \end{displaymath}
  \addtocounter{equation}{-1}
  it can be concluded that
  \begin{displaymath}
    \max_{1\le i,j\le 2}|\hat{a}_{ij}|<\sqrt{(2^{h+1})^2+(2^{h+1})^2}=\sqrt{2}\cdot 2^{h+1},
  \end{displaymath}
  \addtocounter{equation}{-1}
  and therefore
  \begin{equation}
    r_{11}=\max_{1\le j\le 2}\|\hat{a}_j\|_F<\sqrt{(\sqrt{2}\cdot 2^{h+1})^2+(\sqrt{2}\cdot 2^{h+1})^2}=2\cdot 2^{h+1}=2^{h+2},
    \label{e:27}
  \end{equation}
  so $\sigma_{11}'\ll 2\cdot 2^{h+2}=2^{h+3}$, where the right hand
  side is the immediate successor (that represents $\infty$) of the
  largest finite floating-point number, as claimed.
\end{proof}

Using Theorem~\ref{t:1}, from~\eqref{e:12} and \eqref{e:11} then
follows
\begin{equation}
  \Sigma=2^{-s}\Sigma',\quad
  U^{\ast}=U_{\varphi}^{\ast} U_+^{\ast},\quad
  V=V_+^{} V_{\psi}^{},\quad
  U=(U^{\ast})^{\ast},
  \label{e:28}
\end{equation}
where $\Sigma'$ has to be backscaled to obtain the singular values of
the original input matrix.  However, the backscaling should be skipped
if it would cause the singular values to overflow or (less
catastrophic but still inaccurate outcome) underflow, while informing
the user of such an event by preserving the value of $s$.

Specifically, by matrix multiplication, from~\eqref{e:28}
and~\eqref{e:11} it follows
\begin{equation}
  \begin{aligned}
    U&=\cos\alpha\cos\varphi\,P_r
    \begin{bmatrix}
      d_{11}^{}(1-\hat{d}_{22}\tan\alpha\tan\varphi) & d_{11}^{}(\tan\varphi+\hat{d}_{22}^{}\tan\alpha)\\
      -d_{22}^{}(\tan\alpha+\hat{d}_{22}\tan\varphi) & d_{22}^{}(\hat{d}_{22}^{}-\tan\alpha\tan\varphi)
    \end{bmatrix},\\
    V&=\cos\psi\,P_c
    \begin{bmatrix}
      1 & \tan\psi\\
      -\tilde{d}_{22}^{}\tan\psi & \tilde{d}_{22}^{}
    \end{bmatrix}.
  \end{aligned}
  \label{e:29}
\end{equation}
Computing each element of a complex $U$ requires only one complex
multiplication.

Writing a complex number $z$ as
$(\Re{z},\Im{z})=\Re{z}+\mathrm{i}\cdot\Im{z}$, noting that $d_{11}$,
$d_{22}$, and $\hat{d}_{22}$ are the complex conjugates of
$(D^{\ast})_{11}$, $(D^{\ast})_{22}$, and $(\widehat{D}^{\ast})_{22}$,
respectively, and precomputing $c=\cos\alpha\cos\varphi$ and
$t=-\tan\alpha\tan\varphi$, \eqref{e:29} can be expanded as, e.g.,
\begin{displaymath}
  \begin{aligned}
    (P_r^{\ast}U)_{11}^{}&=\hphantom{-}c(d_{11}^{}\cdot(\fma(\Re{\hat{d}_{22}^{}},t,1),\Im{\hat{d}_{22}^{}}t)),\\
    (P_r^{\ast}U)_{21}^{}&=-c(d_{22}^{}\cdot(\fma(\Re{\hat{d}_{22}^{}},\tan\varphi,\tan\alpha),\Im{\hat{d}_{22}^{}}\tan\varphi)),\\
    (P_r^{\ast}U)_{12}^{}&=\hphantom{-}c(d_{11}^{}\cdot(\fma(\Re{\hat{d}_{22}^{}},\tan\alpha,\tan\varphi),\Im{\hat{d}_{22}^{}}\tan\alpha)),\\
    (P_r^{\ast}U)_{22}^{}&=\hphantom{-}c(d_{22}^{}\cdot(\fma(-\tan\alpha,\tan\varphi,\Re{\hat{d}_{22}^{}}),\Im{\hat{d}_{22}^{}})),
  \end{aligned}
\end{displaymath}
\addtocounter{equation}{-1}
but another mathematically equivalent computation that minimizes the
number of roundings required for forming the elements of $P_r^{\ast}U$
as this one does is also valid.
%
%
\section{Vector-friendly data layout}\label{s:3}
%
%
Vectors replace scalars in the SIMD arithmetic operations.  A vector
should hold \texttt{S} elements from the same matrix sequence, with
the same row and column indices, and the consecutive bracketed
indices.  When computing with complex numbers, however, it is more
efficient to keep the real and the imaginary parts of the elements in
separate vectors, since there are no hardware vector instructions for
the complex multiplication and division, e.g., which thus have to be
implemented manually.  Also, a vector should be aligned in memory to a
multiple of \texttt{W} bytes to employ the most efficient versions of
the vector load/store operations.  It is therefore essential to
establish a vector-friendly layout for the matrix sequences
$\mathbf{A}$, $\mathbf{U}$, $\mathbf{V}$, and $\mathbf{\Sigma}'$ in
the linear memory space.  One such layout, inspired by splitting the
real and the complex parts of the matrix elements into separate
vectors~\cite[subsection~6.2]{Singer-DiNapoli-Novakovic-Caklovic-19},
is
\begin{displaymath}
  \settowidth{\fbw}{\fbox{$\Re{a_{ij}^{[\hat{n}]}}$}}
  \Re{\mathbf{a}_{ij}^{}}=\text{\framebox[\fbw][c]{$\Re{a_{ij}^{[1]}}$}\framebox[\fbw][c]{$\Re{a_{ij}^{[2]}}$}\framebox[\fbw][c]{$\vphantom{\Re{a_{ij}^{[\hat{n}]}}}\cdots$}\framebox[\fbw][c]{$\Re{a_{ij}^{[\hat{n}]}}$}}\,,\quad
  \settowidth{\fbw}{\fbox{$\Im{a_{ij}^{[\hat{n}]}}$}}
  \Im{\mathbf{a}_{ij}^{}}=\text{\framebox[\fbw][c]{$\Im{a_{ij}^{[1]}}$}\framebox[\fbw][c]{$\Im{a_{ij}^{[2]}}$}\framebox[\fbw][c]{$\vphantom{\Im{a_{ij}^{[\hat{n}]}}}\cdots$}\framebox[\fbw][c]{$\Im{a_{ij}^{[\hat{n}]}}$}}\,.
\end{displaymath}
\addtocounter{equation}{-1}
where $\Re{\mathbf{a}_{ij}}$ and $\Im{\mathbf{a}_{ij}}$ (similarly,
$\Re{\mathbf{u}_{ij}}$, $\Im{\mathbf{u}_{ij}}$, and
$\Re{\mathbf{v}_{ij}}$, $\Im{\mathbf{v}_{ij}}$) for $i,j\in\{1,2\}$
are the sequences of the real and the imaginary components,
respectively, of the elements in the $i$th row and the $j$th column of
the matrices in $\mathbf{A}$ (similarly, in $\mathbf{U}$ and in
$\mathbf{V}$).

Each train of boxes represents a contiguous region of memory aligned
to \texttt{W} bytes.  In the real case, no $\Im$-boxes exist, but the
layout otherwise stays the same.  The scaled singular values
$\sigma_{\max}^{\prime [k]}$ and $\sigma_{\min}^{\prime [k]}$
from~\eqref{e:12} are stored as
\begin{displaymath}
  \settowidth{\fbw}{\fbox{$\sigma_{\max}^{\prime [\hat{n}]}$}}
  \bm{\sigma}_{\max}'=\text{\framebox[\fbw][c]{$\vphantom{\sigma_{\min}^{\prime [1]}}\sigma_{\max}^{\prime [1]}$}\framebox[\fbw][c]{$\vphantom{\sigma_{\min}^{\prime [2]}}\sigma_{\max}^{\prime [2]}$}\framebox[\fbw][c]{$\vphantom{\sigma_{\min}^{\prime [\hat{n}]}}\cdots$}\framebox[\fbw][c]{$\vphantom{\sigma_{\min}^{\prime [\hat{n}]}}\sigma_{\max}^{\prime [\hat{n}]}$}}\,,\quad
  \settowidth{\fbw}{\fbox{$\sigma_{\min}^{\prime [\hat{n}]}$}}
  \bm{\sigma}_{\min}'=\text{\framebox[\fbw][c]{$\sigma_{\min}^{\prime [1]}$}\framebox[\fbw][c]{$\sigma_{\min}^{\prime [2]}$}\framebox[\fbw][c]{$\vphantom{\sigma_{\min}^{\prime [\hat{n}]}}\cdots$}\framebox[\fbw][c]{$\sigma_{\min}^{\prime [\hat{n}]}$}}\,,
\end{displaymath}
\addtocounter{equation}{-1}
respectively, while the scaling parameters $s^{[k]}$ from~\eqref{e:2}
are laid out as
\begin{displaymath}
  \settowidth{\fbw}{\fbox{$s^{[\hat{n}]}$}}
  \mathbf{s}=\text{\framebox[\fbw][c]{$s^{[1]}$}\framebox[\fbw][c]{$s^{[2]}$}\framebox[\fbw][c]{$\vphantom{s^{[\hat{n}]}}\cdots$}\framebox[\fbw][c]{$s^{[\hat{n}]}$}}\,.
\end{displaymath}
\addtocounter{equation}{-1}

The virtual elements, with their bracketed indices ranging from $n+1$
to $\hat{n}$, serve if present as a (e.g., zero) padding, which
ensures that all vectors, including the last one, formed from the
consecutive elements of a box train, hold the same maximal number
(\texttt{S}) of defined values and can thus be processed in an uniform
manner.

The input sequence $\mathbf{A}$ may initially be in another layout and
has to be repacked before any further computation.  Also, the output
sequences $\mathbf{U}$, $\mathbf{V}$, $\mathbf{\Sigma}'$, and
$\mathbf{s}$ may have to be repacked for a further processing.  Such
reshufflings should be avoided, as they incur a substantial overhead
in both time and memory requirements.

Layout of data, including the intermediate results, in vector
registers during the computation is the same as it is for the box
trains, but with \texttt{S} elements instead of $\hat{n}$.  The
$\mathsf{v}$th vector, for
$1\le\mathsf{v}\le\mathsf{V}=\hat{n}/\mathtt{S}$, encompasses the
consecutive indices $k$,
\begin{equation}
  (\mathsf{v}-1)\cdot\mathtt{S}+1\le k\le\mathsf{v}\cdot\mathtt{S}.
  \label{e:30}
\end{equation}
A vector is loaded into, kept in, and stored from, a variable of the C
type \texttt{\_\_m512d}.

In the following, a bold lowercase letter stands for a vector, and the
uppercase one for a (logical, not necessarily in-memory) matrix
sequence.  For example, $\mathbf{R}$ is a sequence of $\hat{n}$
matrices $R^{[k]}$, of which $\mathbf{R}[\mathsf{v}]$ is a subsequence
of length \texttt{S}, and $\mathbf{r}_{12}[\mathsf{v}]$ is a vector
containing elements $r_{12}^{[k]}$ of $\mathbf{R}^{[k]}$, for some
$\mathsf{v}$ and its corresponding indices $k$ from~\eqref{e:30}.  A
bold constant denotes a vector with all its values being equal to the
given constant.  An arithmetic operation on vectors (or collections
thereof) or matrix sequences is a shorthand for a sequence of the
elementwise operations; e.g.,
\begin{displaymath}
  \mathbf{2}^{-\mathbf{s}}=(2^{-s^{[k]}})_k^{},\quad
  \mathbf{B}\mathbf{C}=(B^{[k]}C^{[k]})_k^{},\quad
  1\le k\le\hat{n}.
\end{displaymath}
\addtocounter{equation}{-1}
where $\mathbf{B}$ and $\mathbf{C}$ are any two matrix sequences,
$B^{[k]}C^{[k]}$ is a product of matrices of order two, and
$\mathbf{2}$ is a collection of vectors with all their values equal to
two.  All bracketed indices are one-based, as it is customary in
linear algebra, but in the C code they are zero-based, being thus one
less than they are in the paper's text.
%
%
\section{Overview of the algorithm for the batched SVDs of order two}\label{s:4}
%
%
When there are two cases, the real and the complex one, all
code-presenting figures cover the latter with a mixture of the actual
statements and a mathematically oriented pseudocode.  The real-case
differences are described in them in the comments starting with
$\mathbb{R}$.  A function name in uppercase, \texttt{NAME}, is a
shorthand for the \texttt{\_mm512\_\emph{name}\_pd} C compiler
intrinsic, if the operation is available in the machine's instruction
set, or for an equivalent sequence of the bit-pattern preserving casts
to and from the integer vectors and an equivalent integer
\texttt{NAME} operation.  More precisely, for
$\mathtt{NAME}\in\{\mathtt{AND},\mathtt{ANDNOT},\mathtt{OR},\mathtt{XOR}\}$
bitwise operations, if the AVX-512DQ instruction set extensions are
not supported, an exception to the naming rule holds:
\begin{displaymath}
  \begin{aligned}
    \mathop{\text{\texttt{NAME}}}(x,y)&=\mathop{\text{\texttt{\_mm512\_castsi512\_pd}}}(\mathop{\text{\texttt{\_mm512\_\emph{name}\_epi64}}}(\hat{x},\hat{y})),\\
    \hat{x}&=\mathop{\mathtt{\_mm512\_castpd\_si512}}(x),\quad\hat{y}=\mathop{\mathtt{\_mm512\_castpd\_si512}}(y).
  \end{aligned}
\end{displaymath}
\addtocounter{equation}{-1}
All other required operations have a \texttt{\_pd} variant
(with \texttt{double} precision vector lanes) in the core AVX-512F
instruction set, so it suffices for implementing the entire
algorithm.  Additionally, let \texttt{CMPLT\_MASK} stand for the
\texttt{\_mm512\_cmplt\_pd\_mask} intrinsic, i.e., for the less-than
lane-wise comparison of two vectors.

The four phases of the algorithm for the batched SVDs of order two, as
listed in section~\ref{s:1}, can be succinctly depicted by the
following logical execution pipeline,
\begin{displaymath}
  \begin{array}{c}
    \begin{array}{ccccccccccccccccccc}
      \mathbf{A}&\longrightarrow&\mathbf{\widehat{A}}&\longrightarrow&\mathbf{A}'&\longrightarrow&\mathbf{A}''&\longrightarrow&\mathbf{A}'''&\longrightarrow&\mathbf{R}''&\longrightarrow&\mathbf{R}'&\longrightarrow&\mathbf{R}&\longrightarrow&\mathbf{\Sigma}'&\xrightarrow[\text{safe}]{\text{\textbf{not}}}&\mathbf{\Sigma}\\[-12pt]
      &\downarrow&&\downarrow&&\downarrow&&\downarrow&&\downarrow&&\downarrow&&\downarrow&&\downarrow\\
      &\mathbf{s}&&\hphantom{\null_c}\mathbf{P}_c^{}&&\hphantom{\null_r^{\ast}}\mathbf{P}_r^{\ast}&&\hphantom{\null^{\ast}}\mathbf{D}^{\ast}&&\hphantom{\null_{\bm{\alpha}}^{\ast}}\mathbf{Q}_{\bm{\alpha}}^{\ast}&&\mathbf{\widetilde{D}}&&\hphantom{\null^{\ast}}\mathbf{\widehat{D}}^{\ast}&&\null\;\mathbf{U}_{\bm{\varphi}}^{\ast},\mathbf{V}_{\bm{\psi}}^{}&\rightarrow&\mathbf{U},\mathbf{V}
    \end{array}\notag\\[15pt]
    \mathbf{U}=\mathbf{P}_r\mathbf{D}\mathbf{Q}_{\bm{\alpha}}\mathbf{\widehat{D}}\mathbf{U}_{\bm{\varphi}},\quad
    \mathbf{V}=\mathbf{P}_c\mathbf{\widetilde{D}}\mathbf{V}_{\bm{\psi}};\quad
    \mathbf{\Sigma}=\mathbf{2}^{-\mathbf{s}} \mathbf{\Sigma}',
  \end{array}
\end{displaymath}
\addtocounter{equation}{-1}
where the first row shows the transformations of $\mathbf{A}$, the
second row contains the various matrix sequences that are the
``by-products'' of the computation, described in section~\ref{e:2},
ending with the sequences of the left and the right singular vectors,
that are formed as indicated in the third row.  As the singular values
$\mathbf{\Sigma}$ can overflow due to the backscaling (see
subsection~\ref{ss:2.3}) of the scaled ones ($\mathbf{\Sigma}'$),
computing them unconditionally is unsafe, and such postprocessing is
left to the user's discretion.  In certain use-cases it might be known
in advance that the singular values cannot overflow/underflow, e.g.,
if the initial matrices have already been well-scaled at their
formation.  The backscaling, performed as in Fig.~\ref{l:1}, is then
unconditionally safe.
\begin{figure}[h!btp]
  \begin{lstlisting}
$-\mathbf{0}$ = SET1(-0.0); // a constant vector with all lanes equal to -0.0
$-\mathbf{s}[\mathsf{v}]$ = XOR($\mathbf{s}[\mathsf{v}]$, $-\mathbf{0}$); // negation is performed as XOR-ing with $-0$
$\bm{\sigma}_{\max}^{}[\mathsf{v}]$ = SCALEF($\bm{\sigma}_{\max}'[\mathsf{v}]$, $-\mathbf{s}[\mathsf{v}]$); $\bm{\sigma}_{\min}^{}[\mathsf{v}]$ = SCALEF($\bm{\sigma}_{\min}'[\mathsf{v}]$, $-\mathbf{s}[\mathsf{v}]$);
$\mathbf{s}[\mathsf{v}]$ = $-\mathbf{0}$; // inform the user that the backscaling has been performed
  \end{lstlisting}
  \caption{Optional vectorized backscaling of $\mathbf{\Sigma}'$ to
    $\mathbf{\Sigma}$ by $\mathbf{2}^{-\mathbf{s}}$.}
  \label{l:1}
\end{figure}

The pipeline is executed independently on each non-overlapping
subsequence of \texttt{S} consecutive matrices.  If there are more
such sequences than the active threads, at a point in time some
sequences might have already been processed, while the others are
still waiting, either for the start or the completion of the
processing.  A conceptual core of a driver routine implementing such a
pass over the data is shown in Fig.~\ref{l:2}, where \texttt{xSsvd2},
$\mathtt{x}\in\{\mathtt{d},\mathtt{z}\}$, are the main (real or
complex, respectively), single-threaded routines that are responsible
for all vectorized computations on each particular sequence of size
\texttt{S}.  The OpenMP~\cite{OpenMP-18} \texttt{parallel for}
directive in Fig.~\ref{l:2} assumes a user-defined maximal number and
placement/affinity of threads.
\begin{figure}[h!btp]
  \begin{lstlisting}
const size_t $\mathsf{V}$ = ($n$ + (S - 1)) / S; // $\mathsf{V}=\left\lceil n/\mathtt{S}\right\rceil$, $\hat{n}=\mathsf{V}\cdot\mathtt{S}$
#pragma omp parallel for shared($\mathsf{V},\mathbf{A},\mathbf{U},\mathbf{V},\mathbf{\Sigma}',\mathbf{s}$)
for (size_t $\mathsf{v}$ = 0; $\mathsf{v}$ < $\mathsf{V}$; ++$\mathsf{v}$) xSsvd2($\mathbf{A}[\mathsf{v}],\mathbf{U}[\mathsf{v}],\mathbf{V}[\mathsf{v}],\mathbf{\Sigma}'[\mathsf{v}],\mathbf{s}[\mathsf{v}]$);
  \end{lstlisting}
  \caption{A conceptualization of the main part of a driver routine
    for the batched SVDs of order two, with an OpenMP parallel loop
    over the data, where each of the $\mathsf{V}$ subsequences of
    length \texttt{S} can be processed concurrently with others by an
    \texttt{xSsvd2} routine that performs \texttt{S} SVDs
    simultaneously.}
  \label{l:2}
\end{figure}

The input arguments of \texttt{xSsvd2} are (the pointers to) the
arrays, each aligned to \texttt{W} bytes, of \texttt{S}
\texttt{double} values, e.g., \texttt{const double A12r[static S]} for
$\Re{\mathbf{a}_{12}}[\mathsf{v}]$.  The output arguments are similar,
e.g., \texttt{double U21i[static S]} for
$\Im{\mathbf{u}_{21}}[\mathsf{v}]$.  Note that the same interface, up
to replacing \texttt{S} by \texttt{1}, would be applicable to the
pointwise \texttt{x1svd2} routine for a single $2\times 2$ SVD, but
without the implied alignment restriction.

No branching is involved explicitly in the \texttt{xSsvd2} routines.
It is therefore fully branch-free, if the used SVML routines are.  All
data, once loaded from memory or computed, is intended to be held in
the \texttt{zmm} vector registers until the output has been formed and
written back to RAM\@.  This goal is almost achievable in the test
setting, since there are two vector register spillages, with a total
of only four extra memory accesses (two writes and two reads), as
reported by the optimizer.  A hand-tuned self-contained assembly might
do away with these as well.

The first three phases of the algorithm are vectorized as described in
sections~\ref{s:5}, \ref{s:6}, and~\ref{s:7}, respectively, since each
of the phases can be viewed as an algorithm on its own.  They are,
however, chained by the dataflow, each having as its input the output
of the previous one.  Should the output of a phase be made available
alongside the final results, it could be written to an additional
memory buffer in the same layout as presented in section~\ref{s:3}.
Otherwise, the intermediate results are not preserved.

Vectorization of the last, fourth phase of the algorithm
from~\eqref{e:29} is as tedious and uninformative as it is
straightforward, and so it is omitted for brevity.  It suffices to say
that $\Re{\mathbf{u}}_{ij}$ (and $\Im{\mathbf{u}}_{ij}$) and
$\Re{\mathbf{v}}_{ij}$ (and $\Im{\mathbf{v}}_{ij}$), for
$1\le i,j\le 2$, are computed from~\eqref{e:29}, using the $\fma$
operation where possible, and~\eqref{e:5} for the complex
multiplications.  The final row permutations by $\mathbf{P}_r$ or
$\mathbf{P}_c$ are performed in the same way as the row swaps in the
URV factorization phase, described in Fig.~\ref{l:6} in
section~\ref{s:6}.  An interested reader is referred to the actual
code in the supplementary material\footnote{Supplementary material is
  available in \texttt{https://github.com/venovako/VecKog}
  repository.}.
%
%
\section{Vectorized exact scalings of the input matrices}\label{s:5}
%
%
Computation of the scaling parameters $\mathbf{s}$ is remarkably
simple, as shown in Fig.~\ref{l:3}.
\begin{figure}[h!btp]
  \begin{lstlisting}
$\mathbf{h}$ = SET1((double)(DBL_MAX_EXP-3)); // set each lane of $\mathbf{h}$ to $h$
$\mathbf{e}_{ij}^{\Re}[\mathsf{v}]$ = SUB($\mathbf{h}$, GETEXP($\Re{\mathbf{a}_{ij}^{}}[\mathsf{v}]$)); // $\mathbf{h}-\exp_2^{}\Re{\mathbf{a}_{ij}^{}}[\mathsf{v}]$
// take $\mathbf{e}^{\Re}[\mathsf{v}]=\min\{\mathbf{e}_{11}^{\Re}[\mathsf{v}],\mathbf{e}_{21}^{\Re}[\mathsf{v}],\mathbf{e}_{12}^{\Re}[\mathsf{v}],\mathbf{e}_{22}^{\Re}[\mathsf{v}]\}$ by a two-level $\min$-reduction
$\mathbf{e}^{\Re}[\mathsf{v}]$ = MIN(MIN($\mathbf{e}_{11}^{\Re}[\mathsf{v}]$, $\mathbf{e}_{21}^{\Re}[\mathsf{v}]$), MIN($\mathbf{e}_{12}^{\Re}[\mathsf{v}]$, $\mathbf{e}_{22}^{\Re}[\mathsf{v}]$));
// $\mathbf{e}_{ij}^{\Im}[\mathsf{v}]$, with $1\le i,j\le 2$, and $\mathbf{e}^{\Im}[\mathsf{v}]$ are computed analogously from $\Im{\mathbf{a}_{ij}^{}}[\mathsf{v}]$
$\mathbf{s}[\mathsf{v}]$ = MIN(SET1(DBL_MAX), MIN($\mathbf{e}^{\Re}[\mathsf{v}]$, $\mathbf{e}^{\Im}[\mathsf{v}]$)); // from $\text{\eqref{e:2}}$, $\mathbb{R}\colon\mathbf{e}^{\Im}[\mathsf{v}]$ nonexistent
$\Re{\mathbf{\hat{a}}_{ij}^{}}[\mathsf{v}]$ = SCALEF($\Re{\mathbf{a}_{ij}^{}}[\mathsf{v}]$, $\mathbf{s}[\mathsf{v}]$); $\Im{\mathbf{\hat{a}}_{ij}^{}}[\mathsf{v}]$ = SCALEF($\Im{\mathbf{a}_{ij}^{}}[\mathsf{v}]$, $\mathbf{s}[\mathsf{v}]$);
  \end{lstlisting}
  \caption{Vectorized computation of the scaling parameters $\mathbf{s}$ from~\eqref{e:2} and the scaling of $\mathbf{A}$.}
  \label{l:3}
\end{figure}
It is advantageous to have
$\mathop{\mathtt{GETEXP}}(\mathbf{a})=\exp_2^{}(\mathbf{a})$ and
$\mathop{\mathtt{SCALEF}}(\mathbf{a},\mathbf{b})=\mathbf{a}\cdot\mathbf{2}^{\mathbf{b}}$
vector operations, returning a correct (or correctly rounded,
respectively) result, even with subnormal inputs (the former) or
outputs (the latter).  Should they not be available on another
platform, their scalar variants (\texttt{frexp} and \texttt{scalbn},
respectively) might be used instead on the values in each lane,
slowing the execution considerably.

Once $\mathbf{\widehat{A}}$ is obtained from $\mathbf{A}$, the column
norms of the former are computed as in Fig.~\ref{l:4}.  Observe that
$\mathop{\mathtt{ABS}}(\mathbf{b})=\mathop{\mathtt{ANDNOT}}(-\mathbf{0},\mathbf{b})$,
since
$\mathop{\mathtt{ANDNOT}}(\mathbf{a},\mathbf{b})=\neg\mathbf{a}\wedge\mathbf{b}$
bitwise, and that having a vectorized \texttt{HYPOT} is essential
here.  Should it not be available, it would have to be carefully
implemented to avoid the overflows in the intermediate results.  A
na\"{\i}ve per-lane computation of $c=\hypot(a,b)$, where $a$ and $b$
are finite, without adjusting the exponents of $a$ and $b$, but with
one extra division instead, is to let $a'=\max\{|a|,|b|\}$,
$b'=\min\{|a|,|b|\}$, $a^+=\max\{a',\mathtt{DBL\_TRUE\_MIN}\}>0$,
$q^+=b'/a^+\le 1$, and $c=a'\cdot\sqrt{\fma(q^+,q^+,1)}$.

\begin{figure}[h!btp]
  \begin{lstlisting}
$|\mathbf{\hat{a}}_{ij}|[\mathsf{v}]$ = HYPOT($\Re{\mathbf{\hat{a}}_{ij}[\mathsf{v}]}$, $\Im{\mathbf{\hat{a}}_{ij}[\mathsf{v}]}$); // from $\text{\eqref{e:1}}$, $\mathbb{R}\colon|\mathbf{\hat{a}}_{ij}|[\mathsf{v}]$ = ANDNOT($-\mathbf{0}$, $\Re{\mathbf{\hat{a}}_{ij}}[\mathsf{v}]$)
$\|\mathbf{\hat{a}}_j\|_F[\mathsf{v}]$ = HYPOT($|\mathbf{\hat{a}}_{1j}|[\mathsf{v}]$, $|\mathbf{\hat{a}}_{2j}|[\mathsf{v}]$); // with $j\in\{1,2\}$, from $\text{\eqref{e:3}}$
  \end{lstlisting}
  \caption{Vectorized computation of the column norms $\|\mathbf{\hat{a}}_j\|_F$ from~\eqref{e:3}.}
  \label{l:4}
\end{figure}
%
%
\section{Vectorized URV factorizations of order two}\label{s:6}
%
%
Having its column norms computed, $\mathbf{\widehat{A}}$ has to be
pivoted, each matrix by a column-swapping permutation (or identity, if
a swap is not required), such that a column with the largest norm
becomes the first one.  This is accomplished in Fig.~\ref{l:5} by the
\texttt{MASK\_BLEND} operation, that selects a value for the $\ell$th
output lane from the same lane in either the first or the second
argument vector, according to a bit-mask \textsc{c} that compactly
encodes the results of the lane-wise $<$-comparisons of the norms by
the \texttt{CMPLT\_MASK} operation.  If the $\ell$th bit in the mask
is zero (i.e., the $\ell$th comparison is false), the $\ell$th output
lane gets its value from the first vector, and the corresponding
permutation $\mathbf{P}_c^{[k]}$, where
$k=(\mathsf{v}-1)\cdot\mathtt{S}+\ell$, is identity; else, the output
value is taken from the second vector, and the permutation encodes a
swap.  All norms are finite and thus ordered, so the complement of the
relation $<$ is $\ge$.
\begin{figure}[h!btp]
  \begin{lstlisting}
$\text{\textsc{c}}[\mathsf{v}]$ = CMPLT_MASK($\|\mathbf{\hat{a}}_1\|_F[\mathsf{v}]$, $\|\mathbf{\hat{a}}_2\|_F[\mathsf{v}]$); // $\mathtt{S}$-bit mask encodes the $<$ relation
$\Re{\mathbf{a}_{i1}'}[\mathsf{v}]$ = MASK_BLEND($\text{\textsc{c}}[\mathsf{v}]$, $\Re{\mathbf{\hat{a}}_{i1}^{}}[\mathsf{v}]$, $\Re{\mathbf{\hat{a}}_{i2}^{}}[\mathsf{v}]$); // similarly for $\Im{\mathbf{a}_{i1}'}[\mathsf{v}]$
$\Re{\mathbf{a}_{i2}'}[\mathsf{v}]$ = MASK_BLEND($\text{\textsc{c}}[\mathsf{v}]$, $\Re{\mathbf{\hat{a}}_{i2}^{}}[\mathsf{v}]$, $\Re{\mathbf{\hat{a}}_{i1}^{}}[\mathsf{v}]$); // similarly for $\Im{\mathbf{a}_{i2}'}[\mathsf{v}]$
$|\mathbf{a}_{i1}'|[\mathsf{v}]$ = MASK_BLEND($\text{\textsc{c}}[\mathsf{v}]$, $|\mathbf{\hat{a}}_{i1}^{}|[\mathsf{v}]$, $|\mathbf{\hat{a}}_{i2}^{}|[\mathsf{v}]$);
$|\mathbf{a}_{i2}'|[\mathsf{v}]$ = MASK_BLEND($\text{\textsc{c}}[\mathsf{v}]$, $|\mathbf{\hat{a}}_{i2}^{}|[\mathsf{v}]$, $|\mathbf{\hat{a}}_{i1}^{}|[\mathsf{v}]$);
$\|\mathbf{a}_1'\|_F^{}[\mathsf{v}]$ = MASK_BLEND($\text{\textsc{c}}[\mathsf{v}]$, $\|\mathbf{\hat{a}}_1^{}\|_F^{}[\mathsf{v}]$, $\|\mathbf{\hat{a}}_2^{}\|_F^{}[\mathsf{v}]$);
$\|\mathbf{a}_2'\|_F^{}[\mathsf{v}]$ = MASK_BLEND($\text{\textsc{c}}[\mathsf{v}]$, $\|\mathbf{\hat{a}}_2^{}\|_F^{}[\mathsf{v}]$, $\|\mathbf{\hat{a}}_1^{}\|_F^{}[\mathsf{v}]$);
  \end{lstlisting}
  \caption{Vectorized column pivoting of $\mathbf{\widehat{A}}$.}
  \label{l:5}
\end{figure}

Not only $\mathbf{A}'$ itself has to be obtained.  The absolute values
of the elements and the column norms also have to be subject to the
same (maybe identity) permutations, as in Fig.~\ref{l:5}, to avoid
recomputing them unnecessarily and at a greater cost, especially in
the complex case.  The similar principles hold for the row sorting of
$\mathbf{A}'$ in Fig.~\ref{l:6}.
\begin{figure}[h!btp]
  \begin{lstlisting}
$\text{\textsc{r}}[\mathsf{v}]$ = CMPLT_MASK($|\mathbf{a}_{11}'|[\mathsf{v}]$, $|\mathbf{a}_{21}'|[\mathsf{v}]$); // Is $|\mathbf{a}_{11}'|[\mathsf{v}]<|\mathbf{a}_{21}'|[\mathsf{v}]$, lane-wise?
$\Re{\mathbf{a}_{1j}''}[\mathsf{v}]$ = MASK_BLEND($\text{\textsc{r}}[\mathsf{v}]$, $\Re{\mathbf{a}_{1j}'}[\mathsf{v}]$, $\Re{\mathbf{a}_{2j}'}[\mathsf{v}]$); // similarly for $\Im{\mathbf{a}_{1j}''}[\mathsf{v}]$
$\Re{\mathbf{a}_{2j}''}[\mathsf{v}]$ = MASK_BLEND($\text{\textsc{r}}[\mathsf{v}]$, $\Re{\mathbf{a}_{2j}'}[\mathsf{v}]$, $\Re{\mathbf{a}_{1j}'}[\mathsf{v}]$); // similarly for $\Im{\mathbf{a}_{2j}''}[\mathsf{v}]$
$|\mathbf{a}_{1j}''|[\mathsf{v}]$ = MASK_BLEND($\text{\textsc{r}}[\mathsf{v}]$, $|\mathbf{a}_{1j}'|[\mathsf{v}]$, $|\mathbf{a}_{2j}'|[\mathsf{v}]$);
$|\mathbf{a}_{2j}''|[\mathsf{v}]$ = MASK_BLEND($\text{\textsc{r}}[\mathsf{v}]$, $|\mathbf{a}_{2j}'|[\mathsf{v}]$, $|\mathbf{a}_{1j}'|[\mathsf{v}]$);
  \end{lstlisting}
  \caption{Vectorized row sorting of $\mathbf{A}'$.}
  \label{l:6}
\end{figure}

Since only the rows of $\mathbf{A}'$ are possibly swapped to get
$\mathbf{A}''$, the column norms do not change, so
$\|\mathbf{a}_j''\|_F^{}=\|\mathbf{a}_j'\|_F^{}$.  To make the first
columns of $\mathbf{A}''$ real and non-negative, $\mathbf{D}^{\ast}$
from~\eqref{e:4} or~\eqref{e:10} is computed and applied as in
Fig.~\ref{l:7}.  Observe how the sign extractions and the implicit
complex conjugations and multiplications are performed in
Fig.~\ref{l:7}, as the same pattern is assumed for them in the
following.
\begin{figure}[h!btp]
  \begin{lstlisting}
$\mathbf{1}$ = SET1(1.0); $\mathbf{m}$ = SET1(DBL_TRUE_MIN); // ones & the successors of $+0$
$|\mathbf{a}_{i1}^+|[\mathsf{v}]$ = MAX($|\mathbf{a}_{i1}''|[\mathsf{v}]$, $\mathbf{m}$); // from $\text{\eqref{e:1}}$, with $i\in\{1,2\}$, here and below
$\Re{\mathbf{d}_{ii}^{}}[\mathsf{v}]$ = OR(MIN(DIV(ANDNOT($-\mathbf{0}$, $\Re{\mathbf{a}_{i1}''}[\mathsf{v}]$), $|\mathbf{a}_{i1}''|[\mathsf{v}]$), $\mathbf{1}$), AND($\Re{\mathbf{a}_{i1}''}[\mathsf{v}]$, $-\mathbf{0}$));
$\Im{\mathbf{d}_{ii}^{}}[\mathsf{v}]$ = DIV($\Im{\mathbf{a}_{i1}''}[\mathsf{v}]$, $|\mathbf{a}_{i1}^+|[\mathsf{v}]$); // $\Im{\mathbf{d}_{ii}^{\ast}}[\mathsf{v}]=-\Im{\mathbf{d}_{ii}^{}}[\mathsf{v}]$ implicitly
//$\,\mathbb{R}\colon\Re{\mathbf{d}_{ii}^{}}[\mathsf{v}]$ = AND($\Re{\mathbf{a}_{i1}''}[\mathsf{v}]$, $-\mathbf{0}$), only the sign bit ($\pm 0$, not $\pm 1$ from $\text{\eqref{e:10}}$)
// $(\Re{\mathbf{a}_{i2}'''}[\mathsf{v}],\Im{\mathbf{a}_{i2}'''}[\mathsf{v}])=(\Re{\mathbf{d}_{ii}^{\ast}}[\mathsf{v}],\Im{\mathbf{d}_{ii}^{\ast}}[\mathsf{v}])\cdot(\Re{\mathbf{a}_{i2}''}[\mathsf{v}],\Im{\mathbf{a}_{i2}''}[\mathsf{v}])$, from $\text{\eqref{e:4}}$
$\Re{\mathbf{a}_{i2}'''}[\mathsf{v}]$ = FMADD($\Re{\mathbf{d}_{ii}^{}}[\mathsf{v}]$, $\Re{\mathbf{a}_{i2}''}[\mathsf{v}]$, MUL($\Im{\mathbf{d}_{ii}^{}}[\mathsf{v}]$, $\Im{\mathbf{a}_{i2}''}[\mathsf{v}]$)); // and below from $\text{\eqref{e:5}}$
$\Im{\mathbf{a}_{i2}'''}[\mathsf{v}]$ = FMSUB($\Re{\mathbf{d}_{ii}^{}}[\mathsf{v}]$, $\Im{\mathbf{a}_{i2}''}[\mathsf{v}]$, MUL($\Im{\mathbf{d}_{ii}^{}}[\mathsf{v}]$, $\Re{\mathbf{a}_{i2}''}[\mathsf{v}]$)); // Fused Mul and SUB
//$\,\mathbb{R}\colon\Re{\mathbf{a}_{i2}'''}[\mathsf{v}]$ = XOR($\Re{\mathbf{d}_{ii}^{}}[\mathsf{v}]$, $\Re{\mathbf{a}_{i2}''}[\mathsf{v}]$), $\Re{\mathbf{a}_{i2}'''}[\mathsf{v}]=\Re{\mathbf{d}_{ii}^{\ast}}[\mathsf{v}]\cdot\Re{\mathbf{a}_{i2}''}[\mathsf{v}]$ from $\text{\eqref{e:10}}$
$\Re{\mathbf{a}_{i1}'''}[\mathsf{v}]$ = $|\mathbf{a}_{i1}''|[\mathsf{v}]$; // assume $\Im{\mathbf{a}_{i1}'''}[\mathsf{v}]$ = SETZERO(); i.e., $\mathbf{0}$
  \end{lstlisting}
  \caption{Vectorized computation of $\mathbf{D}^{\ast}$ and $\mathbf{A}'''$ from~\eqref{e:4} or~\eqref{e:10}.}
  \label{l:7}
\end{figure}

The matrices $\mathbf{D}^{\ast}$ are unitary and diagonal, so
$\|\mathbf{a}_j'''\|_F^{}=\|\mathbf{a}_j''\|_F^{}$ can be (and is)
assumed, though numerically they might slightly differ, should the
former be recomputed, due to the rounding errors accumulated in the
course of the transformations of the elements of $\mathbf{A}''$ as in
Fig.~\ref{l:7}, as well as due to the recomputation itself.

Fig.~\ref{l:8} shows how to get the QR factorizations
from~\eqref{e:6}, i.e.,
$\mathbf{R}''=\mathbf{Q}_{\bm{\alpha}}^{\ast}\mathbf{A}'''$.  Only
$\mathbf{r}_2''$ has to be computed by multiplying  $\mathbf{a}_2'''$
(complex in the general case) from the left by the real
$\mathbf{Q}_{\bm{\alpha}}^{\ast}$, while $\mathbf{r}_1''$ is always
real and already known.  Should \texttt{INVSQRT} not be available,
there are two remedies, both starting from
$\sec\alpha=\sqrt{1+\tan^2\alpha}$.  The first, faster one computes
$\cos\alpha=1/\sec\alpha$, while the second, possibly more accurate
one due to requiring one rounding less than the
first~\cite{Novakovic-Singer-20}, does not require the cosine at all,
and instead replaces all multiplications by it with divisions by the
secant.
\begin{figure}[h!btp]
  \begin{lstlisting}
$\mathbf{r}_{11}^{}[\mathsf{v}]$ = $\|\mathbf{a}_1'''\|_F^{}[\mathsf{v}]$; // from $\text{\eqref{e:6}}$, $\mathbf{r}_{21}^{}[\mathsf{v}]$ = $\mathbf{0}$ is assumed but not set
$-\tan{\bm{\alpha}}[\mathsf{v}]$ = MAX(DIV($\Re{\mathbf{a}_{21}'''}[\mathsf{v}]$, $\Re{\mathbf{a}_{11}'''}[\mathsf{v}]$), $\mathbf{0}$); // from $\text{\eqref{e:7}}$
$\cos{\bm{\alpha}}[\mathsf{v}]$ = INVSQRT(FMADD($\tan{\bm{\alpha}}[\mathsf{v}]$, $\tan{\bm{\alpha}}[\mathsf{v}]$, $\mathbf{1}$)); // from $\text{\eqref{e:7}}$
$\Re{\mathbf{r}_{12}'}[\mathsf{v}]$ = MUL($\cos{\bm{\alpha}}[\mathsf{v}]$, FMADD($-\tan{\bm{\alpha}}[\mathsf{v}]$, $\Re{\mathbf{a}_{22}'''}[\mathsf{v}]$, $\Re{\mathbf{a}_{12}'''}[\mathsf{v}]$)); // similarly $\Im{\mathbf{r}_{12}'}[\mathsf{v}]$
// $\tan{\bm{\alpha}}[\mathsf{v}]$ implicit in Fused_Negative_Multiply-ADD($a,b,c$) = $-(a\cdot b)+c$
$\Re{\mathbf{r}_{22}''}[\mathsf{v}]$ = MUL($\cos{\bm{\alpha}}[\mathsf{v}]$, FNMADD($-\tan{\bm{\alpha}}[\mathsf{v}]$, $\Re{\mathbf{a}_{12}'''}[\mathsf{v}]$, $\Re{\mathbf{a}_{22}'''}[\mathsf{v}]$)); // $\Im{\mathbf{r}_{22}''}[\mathsf{v}]$ likewise
  \end{lstlisting}
  \caption{The vectorized QR factorization of $\mathbf{A}'''$ from~\eqref{e:6} and~\eqref{e:7}.}
  \label{l:8}
\end{figure}

Now $\mathbf{r}_{12}'$ has to be made real and non-negative by
multiplying $\mathbf{R}''$ by $\mathbf{\widetilde{D}}$ from the right,
obtaining $\mathbf{R}'$ as in~\eqref{e:8}, and then $\mathbf{r}_{22}'$
has to undergo a similar procedure by multiplying $\mathbf{R}'$ from
the left by $\mathbf{\widehat{D}}^{\ast}$, as in~\eqref{e:9}, to get
the real and non-negative $\mathbf{R}$.  The first step involves one
complex multiplication per lane,
$\mathbf{r}_{22}'=\mathbf{r}_{22}''\cdot\mathbf{\tilde{d}}_{22}^{}$,
while $\mathbf{r}_{12}^{}=|\mathbf{r}_{12}'|$, and the second step
involves none, since $\mathbf{r}_{22}^{}=|\mathbf{r}_{22}'|$, as shown
in Fig.~\ref{l:9}.
\begin{figure}[h!btp]
  \begin{lstlisting}
$\mathbf{r}_{12}^{}[\mathsf{v}]$ = $|\mathbf{r}_{12}'|[\mathsf{v}]$ = HYPOT($\Re{\mathbf{r}_{12}'}[\mathsf{v}]$, $\Im{\mathbf{r}_{12}'[\mathsf{v}]}$); //$\,\mathbb{R}\colon\mathbf{r}_{12}^{}[\mathsf{v}]$ = ANDNOT($-\mathbf{0}$, $\Re{\mathbf{r}_{12}'}[\mathsf{v}]$)
$\Re{\mathbf{\tilde{d}}_{22}^{\ast}}$ = OR(MIN(DIV(ANDNOT($-\mathbf{0}$, $\Re{\mathbf{r}_{12}'}[\mathsf{v}]$), $|\mathbf{r}_{12}'|[\mathsf{v}]$), $\mathbf{1}$), AND($\Re{\mathbf{r}_{12}'}[\mathsf{v}]$, $-\mathbf{0}$)); // $\text{\eqref{e:8}}$
$\Im{\mathbf{\tilde{d}}_{22}^{\ast}}$ = DIV($\Im{\mathbf{r}_{12}'}[\mathsf{v}]$, MAX($|\mathbf{r}_{12}'|[\mathsf{v}]$, $\mathbf{m}$)); // from $\text{\eqref{e:8}}$, here and above using $\text{\eqref{e:1}}$
//$\,\mathbb{R}\colon\Re{\mathbf{d}_{22}^{}}[\mathsf{v}]$ = AND($\Re{\mathbf{r}_{12}'}[\mathsf{v}]$, $-\mathbf{0}$), from $\text{\eqref{e:10}}$, but only the sign bit (i.e., $\pm 0$)
$\Re{\mathbf{r}_{22}'}[\mathsf{v}]$ = FMADD($\Re{\mathbf{r}_{22}''}[\mathsf{v}]$, $\Re{\mathbf{\tilde{d}}_{22}^{\ast}}[\mathsf{v}]$, MUL($\Im{\mathbf{r}_{22}''}[\mathsf{v}]$, $\Im{\mathbf{\tilde{d}}_{22}^{\ast}}[\mathsf{v}]$)); // and below from $\text{\eqref{e:8}}$
$\Im{\mathbf{r}_{22}'}[\mathsf{v}]$ = FMSUB($\Im{\mathbf{r}_{22}''}[\mathsf{v}]$, $\Re{\mathbf{\tilde{d}}_{22}^{\ast}}[\mathsf{v}]$, MUL($\Re{\mathbf{r}_{22}''}[\mathsf{v}]$, $\Im{\mathbf{\tilde{d}}_{22}^{\ast}}[\mathsf{v}]$)); // Fused Mul and SUB
//$\,\mathbb{R}\colon\Re{\mathbf{r}_{22}'}[\mathsf{v}]$ = XOR($\Re{\mathbf{r}_{22}''}[\mathsf{v}]$, $\Re{\mathbf{d}_{22}^{}}[\mathsf{v}]$), from $\text{\eqref{e:10}}$, but faster than multiplication
$\mathbf{r}_{22}^{}[\mathsf{v}]$ = $|\mathbf{r}_{22}'|[\mathsf{v}]$ = HYPOT($\Re{\mathbf{r}_{22}'}[\mathsf{v}]$, $\Im{\mathbf{r}_{22}'[\mathsf{v}]}$); //$\,\mathbb{R}\colon\mathbf{r}_{22}^{}[\mathsf{v}]$ = ANDNOT($-\mathbf{0}$, $\Re{\mathbf{r}_{22}'}[\mathsf{v}]$)
$\Re{\mathbf{\hat{d}}_{22}^{}}$ = OR(MIN(DIV(ANDNOT($-\mathbf{0}$, $\Re{\mathbf{r}_{22}'}[\mathsf{v}]$), $|\mathbf{r}_{22}'|[\mathsf{v}]$), $\mathbf{1}$), AND($\Re{\mathbf{r}_{22}'}[\mathsf{v}]$, $-\mathbf{0}$)); // $\text{\eqref{e:9}}$
$\Im{\mathbf{\hat{d}}_{22}^{}}$ = DIV($\Im{\mathbf{r}_{22}'}[\mathsf{v}]$, MAX($|\mathbf{r}_{22}'|[\mathsf{v}]$, $\mathbf{m}$)); // $\Im{\mathbf{\hat{d}}_{22}^{\ast}}=-\Im{\mathbf{\hat{d}}_{22}^{}}$ implicitly; from $\text{\eqref{e:9}}$
//$\,\mathbb{R}\colon\Re{\mathbf{\hat{d}}_{22}^{}}[\mathsf{v}]$ = AND($\Re{\mathbf{r}_{22}'}[\mathsf{v}]$, $-\mathbf{0}$), from $\text{\eqref{e:10}}$, but the sign bit extraction only
  \end{lstlisting}
  \caption{Vectorized computation of $\mathbf{\widetilde{D}}$, $\mathbf{\widehat{D}}$, and $\mathbf{R}$ from~\eqref{e:8} and~\eqref{e:9}, or from~\eqref{e:10}.}
  \label{l:9}
\end{figure}
%
%
\section{Vectorized SVD of real upper-triangular matrices of order two}\label{s:7}
%
%
In Fig.~\ref{l:10} a vectorization of the $2\times 2$ SVD method from
subsection~\ref{ss:2.3} for non-negative upper triangular matrices is
shown.  If $\mathtt{T}\ne\mathtt{double}$, the precomputed upper bound
for $\tan(2\varphi)$ should be replaced by the appropriate one (e.g.,
for $\mathtt{T}=\mathtt{float}$, $\sqrt{\mathtt{FLT\_MAX}}$ should be
used instead).  No sines are explicitly computed here, unlike in the
LAPACK's \texttt{DLASV2} routine, but could be, as
$\sin\beta=\cos\beta\cdot\tan\beta$, for $\beta\in\{\varphi,\psi\}$.
\begin{figure}[h!btp]
  \begin{lstlisting}
$\mathbf{f}$ = SET1(1.34078079299425956E+154); // $\sqrt{\mathtt{DBL\_MAX}}$ from $\text{\eqref{e:14}}$
$\mathbf{x}[\mathsf{v}]$ = MAX(DIV($\mathbf{r}_{12}[\mathsf{v}]$, $\mathbf{r}_{11}[\mathsf{v}]$), $\mathbf{0}$); $\mathbf{y}[\mathsf{v}]$ = MAX(DIV($\mathbf{r}_{22}[\mathsf{v}]$, $\mathbf{r}_{11}[\mathsf{v}]$), $\mathbf{0}$); // see $\text{\eqref{e:13}}$
$\tan(\mathbf{2}\bm{\varphi})[\mathsf{v}]$ = OR(MIN(MAX(DIV(MUL(SCALEF(MIN($\mathbf{x}[\mathsf{v}]$, $\mathbf{y}[\mathsf{v}]$), $\mathbf{1}$), MAX($\mathbf{x}[\mathsf{v}]$, $\mathbf{y}[\mathsf{v}]$)),
       /* from $\text{\eqref{e:14}}$ */ FMADD(SUB($\mathbf{x}[\mathsf{v}]$, $\mathbf{y}[\mathsf{v}]$), ADD($\mathbf{x}[\mathsf{v}]$, $\mathbf{y}[\mathsf{v}]$), $\mathbf{1}$)), $\mathbf{0}$), $\mathbf{f}$), $-\mathbf{0}$);
$\tan{\bm{\varphi}}[\mathsf{v}]$ = DIV($\tan(\mathbf{2}\bm{\varphi})[\mathsf{v}]$, ADD($\mathbf{1}$, SQRT(FMADD($\tan(\mathbf{2}\bm{\varphi})[\mathsf{v}]$, $\tan(\mathbf{2}\bm{\varphi})[\mathsf{v}]$, $\mathbf{1}$))));
$\sec^2{\bm{\varphi}}[\mathsf{v}]$ = FMADD($\tan{\bm{\varphi}}[\mathsf{v}]$, $\tan{\bm{\varphi}}[\mathsf{v}]$, $\mathbf{1}$); $\cos{\bm{\varphi}}[\mathsf{v}]$ = INVSQRT($\sec^2{\bm{\varphi}}[\mathsf{v}]$); // see $\text{\eqref{e:15}}$
$\tan{\bm{\psi}}[\mathsf{v}]$ = FMSUB($\mathbf{y}[\mathsf{v}]$, $\tan{\bm{\varphi}}[\mathsf{v}]$, $\mathbf{x}[\mathsf{v}]$); $\sec^2{\bm{\psi}}[\mathsf{v}]$ = FMADD($\tan{\bm{\psi}}[\mathsf{v}]$, $\tan{\bm{\psi}}[\mathsf{v}]$, $\mathbf{1}$);
$\cos{\bm{\psi}}[\mathsf{v}]$ = INVSQRT($\sec^2{\bm{\psi}}[\mathsf{v}]$); // from $\text{\eqref{e:16}}$ here and above
$\mathbf{c}_{\bm{\varphi}}^{\bm{\psi}}[\mathsf{v}]$ = MUL($\cos{\bm{\varphi}}[\mathsf{v}]$, $\cos{\bm{\psi}}[\mathsf{v}]$); $\bm{\sigma}_{\max}'[\mathsf{v}]$ = MUL(MUL($\mathbf{c}_{\bm{\varphi}}^{\bm{\psi}}[\mathsf{v}]$, $\sec^2{\bm{\psi}}[\mathsf{v}]$), $\mathbf{r}_{11}^{}$);
$\bm{\sigma}_{\min}'[\mathsf{v}]$ = MUL(MUL($\mathbf{c}_{\bm{\varphi}}^{\bm{\psi}}[\mathsf{v}]$, $\sec^2{\bm{\varphi}}[\mathsf{v}]$), $\mathbf{r}_{22}^{}$); // from $\text{\eqref{e:17}}$ here and above
  \end{lstlisting}
  \caption{Vectorization of the $2\times 2$ SVD of a non-negative upper-triangular matrix from \eqref{e:13}--\eqref{e:17}.}
  \label{l:10}
\end{figure}

A computation functionally similar to the one proposed in
Fig.~\ref{l:10} could be performed by \texttt{S} calls to
\texttt{DLASV2}.  In the Fortran syntax, one such call looks like
\begin{displaymath}
  \text{\texttt{CALL DLASV2}}(\mathtt{F}^{[k]},\,\mathtt{G}^{[k]},\,\mathtt{H}^{[k]},\,\mathtt{SSMIN}^{[k]},\,\mathtt{SSMAX}^{[k]},\,\mathtt{SNR}^{[k]},\,\mathtt{CSR}^{[k]},\,\mathtt{SNL}^{[k]},\,\mathtt{CSL}^{[k]}),
\end{displaymath}
\addtocounter{equation}{-1}
where $k$ lies in the range~\eqref{e:30}, for a given $\mathsf{v}$.
The input-only arguments are
\begin{displaymath}
  \mathtt{F}^{[k]}=\mathbf{r}_{11}^{[k]},\quad
  \mathtt{G}^{[k]}=\mathbf{r}_{12}^{[k]},\quad
  \mathtt{H}^{[k]}=\mathbf{r}_{22}^{[k]},
\end{displaymath}
\addtocounter{equation}{-1}
while the outputs are related to the quantities computed or implied in
Fig.~\ref{l:10} as
\begin{displaymath}
  \begin{aligned}
    \cos{\bm{\varphi}}[\mathsf{v}]&=(\mathtt{CSL}^{[k]})_k^{},\\
    -\sin{\bm{\varphi}}[\mathsf{v}]&=(\mathtt{SNL}^{[k]})_k^{},
  \end{aligned}\quad
  \begin{aligned}
    \cos{\bm{\psi}}[\mathsf{v}]&=(\mathtt{CSR}^{[k]})_k^{},\\
    -\sin{\bm{\psi}}[\mathsf{v}]&=(\mathtt{SNR}^{[k]})_k^{},
  \end{aligned}\quad
  \begin{aligned}
    \bm{\sigma}_{\max}'[\mathsf{v}]&=(\mathtt{SSMAX}^{[k]})_k^{},\\
    \bm{\sigma}_{\min}'[\mathsf{v}]&=(\mathtt{SSMIN}^{[k]})_k^{}.
  \end{aligned}
\end{displaymath}
\addtocounter{equation}{-1}

It is inadvisable to replace the vectorized algorithm in
Fig.~\ref{l:10} by the \texttt{DLASV2} calls, for at least two
reasons.  First, the input vectors have to be stored from the
registers to the addressable memory.  Then \texttt{S} function calls
have to be made instead of a single pass over the data, and the
results finally have to be loaded from the memory into the vector
registers for the last phase of the algorithm.  Second, throughout the
paper the tangents are used instead of the sines, to increase the
accuracy by reducing the number of the roundings performed due to
more opportunities for employing the $\fma$-type
operations~\cite{Drmac-97}.  However, \texttt{DLASV2} provides the
tangents only implicitly, as $\tan\beta=\sin\beta/\cos\beta$ for
$\beta\in\{\varphi,\psi\}$.  If the last phase of the algorithm comes
after the \texttt{DLASV2} calls, \eqref{e:29} has to be rewritten in
the terms of the respective sines to avoid the superfluous divisions,
as the equivalent expressions
\begin{equation}
  \begin{aligned}
    U&=P_r
    \begin{bmatrix}
      d_{11}^{}(\cos\alpha\cos\varphi-\hat{d}_{22}\sin\alpha\sin\varphi) & d_{11}^{}(\cos\alpha\sin\varphi+\hat{d}_{22}^{}\sin\alpha\cos\varphi)\\
      -d_{22}^{}(\sin\alpha\cos\varphi+\hat{d}_{22}\cos\alpha\sin\varphi) & d_{22}^{}(\hat{d}_{22}^{}\cos\alpha\cos\varphi-\sin\alpha\sin\varphi)
    \end{bmatrix},\\
    V&=P_c
    \begin{bmatrix}
      \cos\psi & \sin\psi\\
      -\tilde{d}_{22}^{}\sin\psi & \tilde{d}_{22}^{}\cos\psi
    \end{bmatrix}.
  \end{aligned}
  \label{e:31}
\end{equation}

A quick test (albeit with computing $\tan\beta$ by two vector
divisions for simplicity) has shown that an algorithm that calls
\texttt{DLASV2} as described is noticeably slower, more so in the real
than in the complex (more involved in the other phases) case,
relatively to the timings of the fully vectorized algorithm.  Using
\texttt{DLASV2} instead of the method in Fig.~\ref{l:10} might
therefore be a viable alternative only in the pointwise case, within a
routine (e.g., \texttt{x1svd2}) designed in the LAPACK's style.
%
%
\section{Numerical testing}\label{s:8}
%
%
All testing was performed on an Intel Xeon Phi 7210 CPU, running at
$1.3$~GHz with TurboBoost turned off in Quadrant cluster mode, with
$96$~GiB of RAM and $16$~GiB of flat-mode MCDRAM (that was not used,
since it is not available on the more recent generations of the Intel
CPUs), under $64$-bit CentOS Linux $7.8.2003$.

The Intel C compiler, \texttt{icc} (version 19.1.1.217), was invoked
with the following optimization and floating-point options:
\texttt{-O3 -xHost -qopt-zmm-usage=high -fp-model source -no-ftz -prec-div -prec-sqrt -fimf-precision=high},
to enable the gradual underflow, prohibit the aggressive
floating-point optimizations that could result in a loss of precision,
and, together with \texttt{-fimf-use-svml=true}, to employ the
high-accuracy SVML library.  Among other options were
\texttt{-std=c18} and \texttt{-qopenmp}.  The \texttt{DLASV2} routine
was provided by the sequential Intel Math Kernel Library (MKL).  The
quadruple precision floating-point arithmetic, used for the error
checking only, was supported by the \texttt{\_\_float128} datatype and
the functions operating on it, e.g., \texttt{\_\_fmaq},
\texttt{\_\_hypotq}, and \texttt{\_\_scalbq}, with the obvious
semantics.

The test data was harvested from \texttt{/dev/urandom} pseudorandom
byte stream, with an approximately uniform probability distribution of
each bit, and 64 consecutive bits formed a \texttt{double} precision
value.  A value that was not finite (either $\pm\infty$ or a
\texttt{NaN}) was replaced with another one that was, sourced in the
same way.  In total $2^{36}$ finite doubles were stored in a binary
file and reused for all runs.  In the real case, the file layout was
assumed to be a sequence of records, where each record contained four
vectors (i.e., all the elements of $\texttt{S}$ $2\times 2$ input
matrices), as
\begin{displaymath}
  \Re{\mathbf{a}_{11}}[\mathsf{v}],\Re{\mathbf{a}_{21}}[\mathsf{v}],\Re{\mathbf{a}_{12}}[\mathsf{v}],\Re{\mathbf{a}_{22}}[\mathsf{v}],
\end{displaymath}
\addtocounter{equation}{-1}
while in the complex case eight vectors were assumed per each record,
as
\begin{displaymath}
  \Re{\mathbf{a}_{11}}[\mathsf{v}],\Im{\mathbf{a}_{11}}[\mathsf{v}],\Re{\mathbf{a}_{21}}[\mathsf{v}],\Im{\mathbf{a}_{21}}[\mathsf{v}],\Re{\mathbf{a}_{12}}[\mathsf{v}],\Im{\mathbf{a}_{12}}[\mathsf{v}],\Re{\mathbf{a}_{22}}[\mathsf{v}],\Im{\mathbf{a}_{22}}[\mathsf{v}].
\end{displaymath}
\addtocounter{equation}{-1}

In both cases, a single batch comprised $n=\hat{n}=2^{28}$
matrices---an absurdly large number for a typical usage scenario in
the $2n\times 2n$ SVD algorithm, but necessary for a reliable timing
of each batch, to compensate for the unavoidable operating system's
jitter.  Therefore, in the real case the test file contained 64
batches, and the same file provided 32 batches in the complex case.
Execution of each batch by the parallel for loop from Fig.~\ref{l:2}
with 32 OpenMP threads, spread across the 64 CPU cores such that each
thread was affinity-bound to its own core while no two threads shared
the same level-2 cache, resulted in the following summary outputs:
\begin{compactenum}[1.]
  \item the wall time $t$ in seconds for processing the entire batch,
    measured by placing the \texttt{omp\_get\_wtime} calls immediately
    before and after the aforesaid parallel loop,
  \item the maximal \emph{a posteriori\/} spectral condition number of
    the matrices in the batch,
    \begin{equation}
      \kappa=\max_{1\le k\le n}\mathop{\kappa_2^{}}(\mathbf{A}^{[k]})=\fmin(\sigma_{\max}^{[k]}/\sigma_{\min}^{[k]},\infty),
      \label{e:32}
    \end{equation}
  \item the maximal normwise relative error of the singular value
    decompositions, as
    \begin{equation}
      \rho=\max_{1\le k\le n}\fmax\big(\|\mathbf{U}^{[k]}\mathbf{\Sigma}^{[k]}(\mathbf{V}^{[k]})^{\ast}-\mathbf{A}^{[k]}\|_F/\|\mathbf{A}^{[k]}\|_F,0\big),
      \label{e:33}
    \end{equation}
  \item the maximal departure from orthogonality of the left singular
    vectors, as
    \begin{equation}
      \delta=\max_{1\le k\le n}\|(\mathbf{U}^{[k]})^{\ast}\mathbf{U}^{[k]}-I\|_F,
      \label{e:34}
    \end{equation}
  \item and the maximal departure from orthogonality of the right singular
    vectors, as
    \begin{equation}
      \eta=\max_{1\le k\le n}\|(\mathbf{V}^{[k]})^{\ast}\mathbf{V}^{[k]}-I\|_F.
      \label{e:35}
    \end{equation}
\end{compactenum}

The last four metrics above ($\kappa$, $\rho$, $\delta$, and $\eta$)
were computed after the batch had been entirely processed and the
output data had been converted to quadruple precision by the
value-preserving casts.  The results were then printed out by
rounding them first to the hardware's 80-bit extended datatype
(\texttt{long double}, with a negligible error), while the timings
were rounded to the nearest microsecond.

The pointwise algorithm was implemented in the real (\texttt{d1svd2})
and the complex (\texttt{z1svd2}) variant.  The first two phases of
the pointwise algorithm are arithmetically equivalent to those of the
vectorized one if the scalar $\hypot$ and $\invsqrt$ functions are
equivalent to the respective vector ones, in an arbitrary lane.  The
SVD of a non-negative upper-triangular matrix was performed in the
pointwise algorithm by a single \texttt{DLASV2} call (see
section~\ref{s:7}), and the subsequent formation of $U$ and $V$ was
done as in~\eqref{e:31}, to compare the accuracy (i.e., $\rho$,
$\delta$, and $\eta$) of such an approach with the one proposed for
the vectorized algorithm.  Also, the pointwise implementations in C
are as close as possible to the LAPACK-style Fortran routines that
could be written for this specific purpose of computing the general
$2\times 2$ SVD, without any of the overhead a call to an $m\times n$
SVD routine would necessarily incur, thus allowing a fair comparison
of the execution times, as follows.  Each call to \texttt{dSsvd2} or
\texttt{zSsvd2} was replaced by \texttt{S} calls (one for each of the
argument vectors' lanes) to \texttt{d1svd2} or \texttt{z1svd2},
respectively, and the rest of the testing code (i.e., the batch timing
and the error checking parts) was left intact.  The speedup is a ratio
of the wall time required for processing a batch with the pointwise
algorithm so employed and the wall time required for the same job
using the vectorized algorithm as proposed.

The maximal condition number $\kappa$ from~\eqref{e:32} attained in
the real case varies from one batch to another, from
$4.447666\cdot 10^{617}$ to $2.134020\cdot 10^{620}$, and in the
complex case from $6.545644\cdot 10^{614}$ to
$9.167483\cdot 10^{616}$, so in each batch there was at least one
almost as highly ill-conditioned matrix as possible, without being
exactly singular.

Fig.~\ref{f:11} shows the attained speedup.  In the real case, the
wall times $t$ vary from $5.767176\,\mathrm{s}$ to
$5.793884\,\mathrm{s}$ for the vectorized algorithm, and from
$21.611207\,\mathrm{s}$ to $21.993982\,\mathrm{s}$ for the pointwise
one.  In the complex case, the ranges are $18.844251\,\mathrm{s}$ to
$19.234680\,\mathrm{s}$, and $54.771144\,\mathrm{s}$ to
$55.227508\,\mathrm{s}$, respectively.  Being somewhat more (in the
real case) or less (in the complex case) than three times, the speedup
has thus justified the purpose of designing the vectorized algorithm,
but also suggests that the pointwise algorithm should be used instead
when $n$ equals one or two.
\begin{figure}[h!btp]
  \begin{center}
    \includegraphics[keepaspectratio]{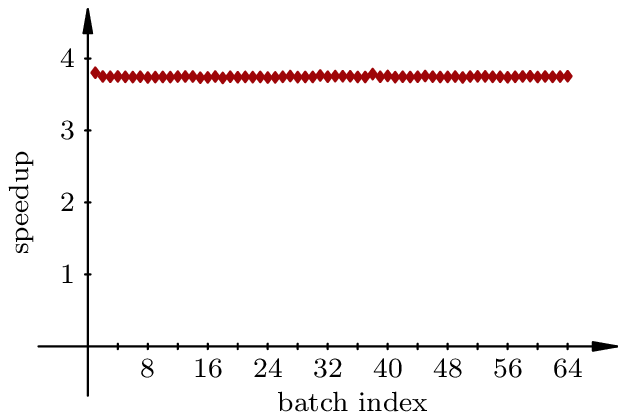}\hfill
    \includegraphics[keepaspectratio]{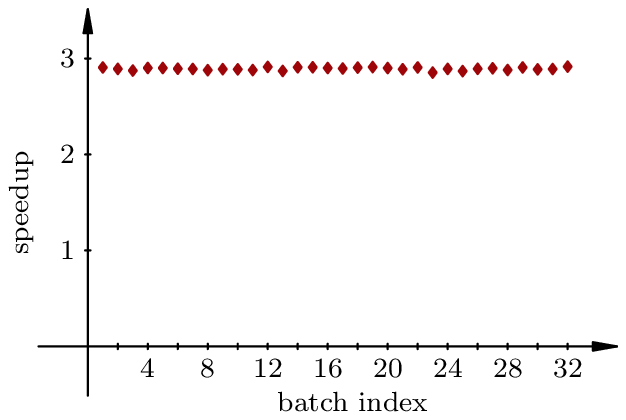}
  \end{center}
  \caption{The speedup attained with each batch processed by the
    vectorized algorithm (\texttt{xSsvd2}) vs.\ the pointwise
    algorithm (\texttt{x1svd2}), in the real ($\mathtt{x}=\mathtt{d}$,
    left) and the complex ($\mathtt{x}=\mathtt{z}$, right) case.}
  \label{f:11}
\end{figure}

Fig.~\ref{f:12} shows that for the vast majority of batches, the
vectorized algorithm gives a bit more accurate decomposition than the
pointwise one, but both are usable.  The optional backscaling was
proven to be unsafe, since in each batch at least one singular value
overflowed when backscaled, with either algorithm and in either case.
\begin{figure}[h!btp]
  \begin{center}
    \includegraphics[keepaspectratio]{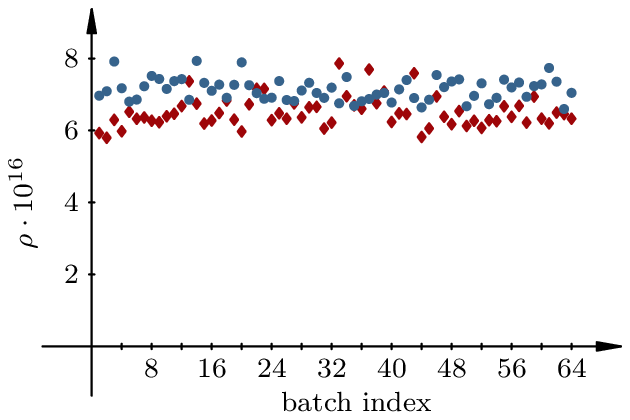}\hfill
    \includegraphics[keepaspectratio]{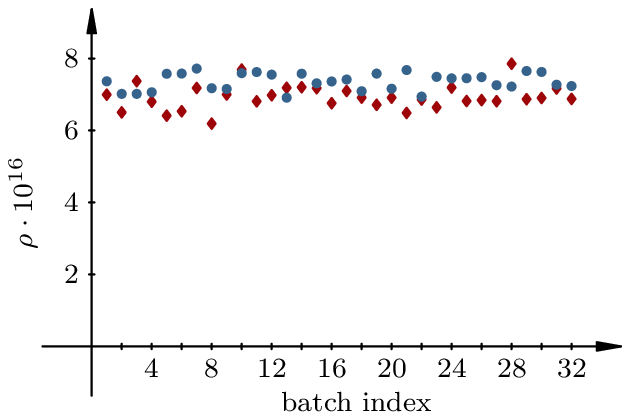}
  \end{center}
  \caption{The maximal normwise relative errors $\rho$
    from~\eqref{e:33} for each batch processed by the vectorized
    ({\color{vec}$\blacklozenge$}) and the pointwise
    ({\color{ptw}$\bullet$}) algorithm, in the real (left) and
    the complex (right) case.}
  \label{f:12}
\end{figure}

Figs.~\ref{f:13} and~\ref{f:14} demonstrate that the vectorized
algorithm generally results in the more orthogonal left and right
singular vectors, respectively, than the pointwise one, since there is
less rounding involved in computing the vectors in the former.
\begin{figure}[h!btp]
  \begin{center}
    \includegraphics[keepaspectratio]{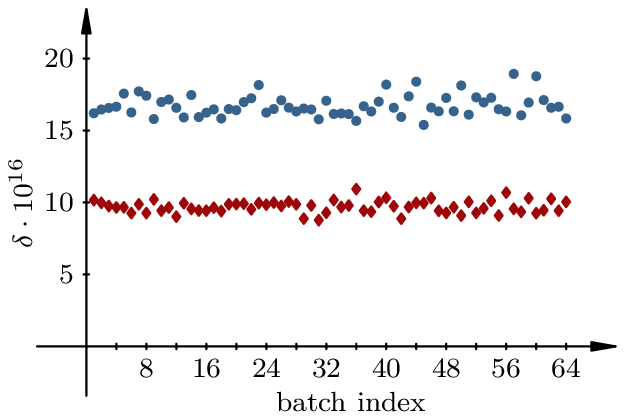}\hfill
    \includegraphics[keepaspectratio]{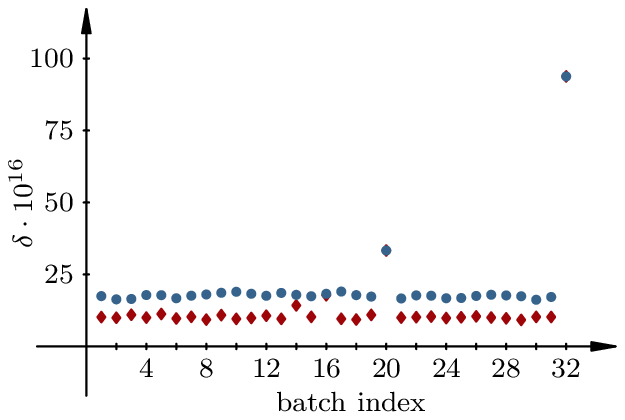}
  \end{center}
  \caption{The maximal departures from orthogonality $\delta$
    from~\eqref{e:34} for each batch processed by the vectorized
    ({\color{vec}$\blacklozenge$}) and the pointwise
    ({\color{ptw}$\bullet$}) algorithm, in the real (left) and
    the complex (right) case.}
  \label{f:13}
\end{figure}
\begin{figure}[h!btp]
  \begin{center}
    \includegraphics[keepaspectratio]{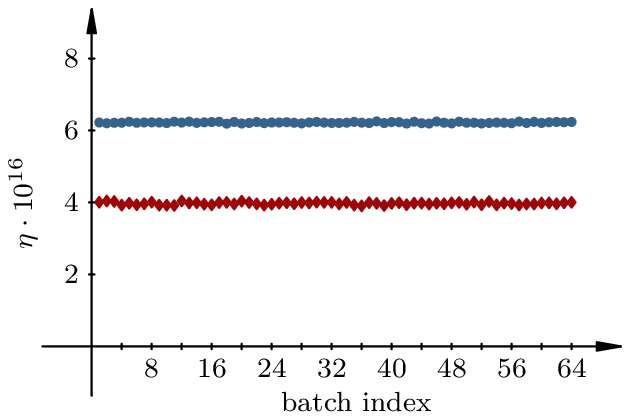}\hfill
    \includegraphics[keepaspectratio]{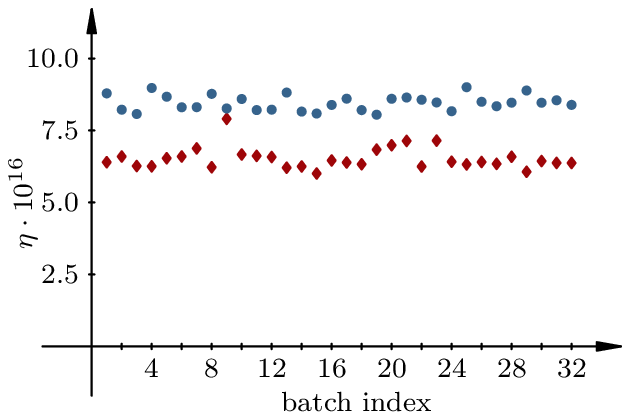}
  \end{center}
  \caption{The maximal departures from orthogonality $\eta$
    from~\eqref{e:35} for each batch processed by the vectorized
    ({\color{vec}$\blacklozenge$}) and the pointwise
    ({\color{ptw}$\bullet$}) algorithm, in the real (left) and
    the complex (right) case.}
  \label{f:14}
\end{figure}
%
%
\section{Conclusions}\label{s:9}
%
%
This paper has shown that a batched computation of the SVDs of order
two can be vectorized with a relative ease on the Intel AVX-512
architecture.  Other vectorization platforms might be targeted as
well, if they provide the instructions analogous to those required
here.  Single precision could be used instead of double precision.

Compared to the pointwise processing of one matrix at a time, the
vectorized algorithm is nearly or more than three times faster, in the
complex and the real case, respectively, and generally slightly more
accurate when computing the SVDs of non-negative upper-triangular
matrices as proposed versus the standard procedure of \texttt{DLASV2}.
Additionally, the exact scalings of the input matrices ensure that the
scaled singular values never overflow, provided that the input
elements are all finite.

A similar vectorization principle, relying on a vector-friendly data
layout (see section~\ref{s:3}), could be applied to the various
algorithms for factorizations or decompositions of a batch of matrices
of a small, fixed order, if the control flow within the algorithm can
be transformed into a branch-free one, as it has been done here for
the column and the row pivotings in section~\ref{s:6}, and for
handling the special values.
%
%
\section*{Acknowledgements}
%
%
This work has been supported in part by Croatian Science Foundation
under the project IP--2014--09--3670
``Matrix Factorizations and Block Diagonalization Algorithms''
(MFBDA).

The author would like to thank Sanja Singer for her assistance in
preparation of Figs.~\ref{f:11}, \ref{f:12}, \ref{f:13},
and~\ref{f:14} with \MP.
%
%

%
%
\end{document}